\documentclass[12pt]{article}

\usepackage[a4paper, total={6in, 8in}]{geometry}

\newenvironment{proof}{\noindent{\bf Proof.}}{\hfill$\Box$\bigskip}

\newcommand{\maxmin}{\mathit{maxmin}}

\newcommand{\wmax}{\mathit{win}}
\newcommand{\lose}{\mathit{lose}}
\newcommand{\pmax}{\mathit{max}}

\newcommand{\s}{}

\usepackage{hyperref}

\usepackage{url}

\usepackage{tikz}
\usetikzlibrary{arrows,backgrounds,positioning}

\usepackage[ruled,vlined,linesnumbered,lined]{algorithm2e}

\usepackage{amsmath}
\usepackage{amsfonts}
\usepackage{amssymb}
\usepackage{latexsym}
\usepackage{alltt}

\usepackage{graphics}
\usepackage{subcaption}

\usepackage{xspace}
\usepackage{graphicx,color}

\usepackage{enumerate}
\usepackage{color}
\usepackage{sgame}

\newtheorem{theorem}{Theorem}%[section]
\newtheorem{defined}[theorem]{Definition}

\newtheorem{exa}{Example}%[section]
\newenvironment{example}{\begin{exa} \rm}{\end{exa}}

\newtheorem{lemma}[theorem]{Lemma}
\newtheorem{corollary}[theorem]{Corollary}

\newtheorem{note}[theorem]{Note}

\newtheorem{exe}{Exercise}%[chapter]

\newtheorem{pro}{Problem}

\newcommand{\bfe}[1]{\begin{bfseries}\emph{#1}\end{bfseries}}

\newcommand{\ES}{\mbox{$\emptyset$}}

\newcommand{\myra}{\mbox{$\:\rightarrow\:$}}

\newcommand{\La}{\mbox{$\:\Leftarrow\:$}}
\newcommand{\Ra}{\mbox{$\:\Rightarrow\:$}}

\newcommand{\sse}{\mbox{$\:\subseteq\:$}}

\newcommand{\fa}{\mbox{$\forall$}}
\newcommand{\te}{\mbox{$\exists$}}

\newcommand{\LL}{\mbox{$\ldots$}}

\newcommand{\NI}{\noindent}
\newcommand{\HB}{\hfill{$\Box$}}

\newcommand{\II}{\vspace{2 mm}}

\newcommand{\szkew}[1]{\relax \setbox0=\hbox{\kern -24pt $\displaystyle#1$\kern 0pt }%
\box0}
{\catcode`\@=11 \global\let\ifjusthvtest@=\iffalse}
\newcounter{oldmycaption}

\newcommand{\turn}{\mathit{turn}}

\newcommand{\leaf}{\mathit{leaf}}
\newcommand{\play}{\mathit{play}}
\newcommand{\win}{\mathit{win}}
\newcommand{\draw}{\mathit{draw}}

\author{Krzysztof R. Apt \\ 
  CWI, Amsterdam, The Netherlands\\
  MIMUW, University of Warsaw, Poland\\
  \url{apt@cwi.nl} 
   \and 
   Sunil Simon\\ 
   Department of CSE, IIT Kanpur, India\\
         \url{simon@cse.iitk.ac.in}
   }

\title{A tutorial for computer scientists on finite extensive games with perfect information} 

\begin{document}

\date{}

\maketitle

\begin{abstract}
  We provide a self-contained introduction to finite extensive games
  with perfect information.  In these games players proceed in turns
  having, at each stage, finitely many moves to their disposal, each
  play always ends, and in each play the players have complete
  knowledge of the previously made moves.  Almost all discussed
  results are well-known, but often they are not presented in an
  optimal form. Also, they usually appear in the literature aimed at
  economists or mathematicians, so the algorithmic or logical aspects
  are underrepresented.
\end{abstract}

\tableofcontents

\section{Introduction}

In computer science for a long time the most commonly studied games
have been infinite two-player games (see, e.g., \cite{AG11} for an
account of some of the most popular classes).  With the advent of
algorithmic game theory various classes of games studied by
economists became subject of interest of computer scientists, as
well. These games usually involve an arbitrary finite number of
players.  Among them one the most common ones are strategic games, in
which the players select their strategies simultaneously. They have
been covered in several books and surveys.

However, in our view a systematic account of another most popular
class of games, extensive games with perfect information, is missing.
It is true that they are extensively discussed in several books,
mostly written for theoretical economists.  However, in the
introductory texts technical results are usually omitted and
illustrated by examples (e.g., \cite{Dut01}).  In turn, in the
advanced textbooks the presentation is often difficult to follow since
these games are introduced as a special case of the extensive games
with imperfect information, which leads to involved notation (e.g.,
\cite{JR01}).  An exception is \cite{OR94} which devotes a separate
chapter to extensive games with perfect information.

From the point of view of computer science the main results are
usually not presented in an optimal way.  For example, the backward
induction is often introduced in a verbose way, or formulated in a way
that hides its algorithmic aspect. This way the optimal result that
relates it to the set of \emph{all} subgame perfect equilibria
(Theorem \ref{thm:extSPE}) is often missed.  We explain here that it
is actually a nondeterministic algorithm that can be even presented as
a parallel algorithm. We also discuss an important article
\cite{Aum91} that formalizes common knowledge of rationality in finite
extensive games with perfect information and relates it to backward
induction.

Further, Theorem \ref{thm:correspondence} that clarifies the relation
between backward induction and the iterated elimination of weakly
dominated strategies, is either only illustrated by an example (e.g.,
\cite{Dut01}) or only a sketch of a proof is provided (e.g.,
\cite{OR94}).  Also, some more recent results, Theorem
\ref{thm:sciewds}, due to \cite{Ewe02}, and Theorem \ref{thm:Kuk}, due
to \cite{Kuk02}, merit in our opinion some attention among computer
scientists.  Finally, the so-called Zermelo theorem about chess-like
games is in our view often proved in a too elaborate way, starting
with the exposition in the classic \cite{vNM04}.

These considerations motivated us to write a tutorial presentation of
finite extensive games with perfect information aimed at computer
scientists.  Often these games are discussed by introducing strategic
games first. We follow this approach, as well, as it allows us to view
extensive games as a subclass of strategic games for which
some additional notions can be defined and for which additional
results hold.

In our presentation we shall often refer to the account given in
\cite{OR94} that comes closest to our ideal.  We shall strengthen some
of their results, provide more detailed proofs of some of them, and
add some new results. Also, we shall recall their natural notion of a
\emph{reduced strategy} that in our view merits more attention.

\section{Preliminaries on strategic games}
\label{sec:prelim}

To discuss extensive games it is convenient to introduce first strategic games.
A \bfe{strategic game} for $n \geq 1$ players consists of:

\begin{itemize}

\item a set of players $\{1, \LL, n\}$,

\end{itemize}
and for each player $i$
\begin{itemize}

\item a non-empty (possibly infinite) set $S_i$ of \bfe{strategies},

\item a \bfe{payoff function} $p_i : S_1 \times \cdots \times S_n \myra \mathbb{R}$.
\end{itemize}

We denote it by $(S_1, \LL, S_n, p_1, \LL, p_n)$. So the set of
players is implicit in this notation.

Strategic games are studied under the following basic assumptions:
\begin{itemize}

\item players choose their strategies \emph{simultaneously}; 
subsequently each player receives a payoff from the
resulting joint strategy, 

\item each player is \bfe{rational},
 which means that his objective is to
 maximize his payoff,

\item players have \bfe{common knowledge} of the game and of each
  others’ rationality.\footnote{Intuitively, common knowledge of some
    fact means that everybody knows it, everybody knows that everybody
    knows it, etc. It is discussed in the context of extensive games
    in Section \ref{sec:ckr}.}

\end{itemize}

Finite two-player games are usually represented in the form called a
\bfe{bimatrix}, where one assumes that the players choose independently
a row or a column. Each entry represents the resulting payoffs to the
row and column players. The following examples of two-player games
will be relevant in the subsequent discussion.

\begin{example} \label{exa:stra}
  
  The following game is called \bfe{Prisoner's Dilemma}. The
  strategies $C$ and $D$ stand for `cooperate' and `defect':
  
\begin{center}
\begin{game}{2}{2}
       & $C$    & $D$\\
$C$   &$2,2$   &$0,3$\\
$D$   &$3,0$   &$1,1$
\end{game}
\end{center}  

The following game is called \bfe{Matching Pennies}. The
strategies $H$ and $T$ stand for `head' and `tail':

\begin{center}
\begin{game}{2}{2}
      & $H$    & $T$\\
$H$   &$\phantom{-}1,-1$   &$-1,\phantom{-}1$\\
$T$   &$-1,\phantom{-}1$   &$\phantom{-}1,-1$
\end{game}
\end{center}
\HB
\end{example}

Fix a strategic game $H:= (S_1, \LL, S_n, p_1, \LL, p_n)$.  Let
$S = S_1 \times \cdots \times S_n$.  We call each element $s \in S$ a
\bfe{joint strategy (of players $1, \LL, n$)}, denote the $i$th element
of $s$ by $s_i$, and abbreviate the sequence $(s_{j})_{j \neq i}$ to
$s_{-i}$. We write $(s'_i, s_{-i})$ to denote the joint strategy in
which player's $i$ strategy is $s'_i$ and each other player's $j$
strategy is $s_j$.  Occasionally we write $(s_i, s_{-i})$ instead of
$s$.  Finally, we abbreviate the Cartesian product
$\times_{j \neq i} S_j$ to $S_{-i}$.

Given a joint strategy $s$, we denote the sequence $(p_1(s), \LL,
p_n(s))$ by $p(s)$ and call it an \bfe{outcome} of the game.  We say
that $H$ \bfe{has $k$ outcomes} if $|\{p(s) \mid s \in S\}|= k$ and
call a game \bfe{trivial} if it has one outcome.  We say that two
joint strategies $s$ and $t$ are \bfe{payoff equivalent} if $p(s) =
p(t)$.

We call a strategy $s_i$ of player $i$ a \bfe{best response} to a 
joint strategy $s_{-i}$ of the other players if
\[
\fa s'_i \in S_i: p_i(s_i, s_{-i}) \geq p_i(s'_i, s_{-i}).
\]

Next, we call a joint strategy $s$ a \bfe{(pure) Nash equilibrium} if for
each player $i$, $s_i$ is a best response to $s_{-i}$, that is, if
\[
\fa i \in \{1, \ldots, n\}: \fa s'_i \in S_i \ p_i(s_i, s_{-i}) \geq p_i(s'_i, s_{-i}).
\]
So a joint strategy is a Nash equilibrium if no player can achieve a
higher payoff by \emph{unilaterally} switching to another
strategy. Intuitively, a Nash equilibrium is a situation in which each
player is a posteriori satisfied with his choice.  It is often used to
predict the outcomes of the strategic games.

It is easy to check that there is a unique Nash equilibrium in
Prisoner's Dilemma which is $(D,D)$, while the Matching Pennies game
has no Nash equilibria.

A relevant question is whether we can identify natural subclasses of
strategic games where a Nash equilibrium is guaranteed to
exist. Below, we describe two such classes which
are well studied in game theory. Fix till the end of this section a
strategic game $H = (S_1, \LL, S_n, p_1, \LL, p_n)$.

We say that a pair of joint strategies $(s, s')$ forms a
\bfe{profitable deviation} if there exists a player $i$ such that
$s'_{-i} = s_{-i}$ and $p_i(s') > p_i(s)$. If such a pair $(s, s')$
exists, then we say that player $i$ can \bfe{profitably deviate} from
$s$ to $s'$ and denote this by $s \to s'$. An \bfe{improvement path}
is a maximal sequence (i.e., a sequence that cannot be extended to the
right) of joint strategies in which each consecutive pair is a
profitable deviation.  By an \bfe{improvement sequence} we mean a
prefix of an improvement path.

We say that $H$ has the \bfe{finite improvement property} ({\bfe{FIP}
  in short), if every improvement path is finite.  Clearly, if an
  improvement path is finite, then its last element is a Nash
  equilibrium. So if $H$ has the FIP, then it is guaranteed to have
  Nash equilibrium, which explains the interest in this notion.  A
  trivial example of a game that has the FIP is the Prisoner's Dilemma
  game.\footnote{A more interesting class of games that have the FIP
    are the congestion games, see \cite{Ros73}.}

However, the FIP is a very strong property and several natural games
with a Nash equilibrium fail to satisfy it.  Young in \cite{You93} and
independently Milchtaich in \cite{Mil96} proposed a weakening of
  this condition and introduced the following natural class of games.
We say that $H$ is \bfe{weakly acyclic} if for any joint strategy $s$,
there exists a finite improvement path that starts at $s$.
Consequently, every weakly acyclic game has a Nash equilibrium.

We call the function $P: S \myra \mathbb{R}$ a \bfe{weak potential}
for $H$ if
\[
\begin{array}{l}
\mbox{$\fa s$: if $s$ is not a Nash equilibrium, then for some} \\
\mbox{profitable deviation $s \to s'$, $P(s) < P(s')$.}
\end{array}
\]

The following natural characterization of finite weakly acyclic games
was established in \cite{Mil13}.

\begin{theorem}[Weakly acyclic] \label{thm:weakly}
A finite game is weakly acyclic iff it has a weak potential.
\end{theorem}

Sometimes it is convenient to assume that a weak potential is a
function to a strict linear ordering (that can be subsequently mapped
to $\mathbb{R}$).

One way to find a Nash equilibrium in a strategic game is by using a
concept of dominance.  In the context of extensive games the most
relevant is the notion of weak dominance.  By a \bfe{subgame} of a
strategic game $H$ we mean a game obtained from $H$ by removing some
strategies.

Consider % a game $H := (S_1, \ldots, S_n, p_1, \ldots, p_n)$ and
two strategies $s_i$ and $s'_i$ of player $i$.  We say that $s_i$
\bfe{weakly dominates $s'_i$} (or equivalently, that $s'_i$ is
\bfe{weakly dominated by $s_i$}) in $H$ if
\[                                                                              
\mbox{$\fa s_{-i} \in S_{-i} : p_i(s_i, s_{-i}) \geq p_i(s'_i, s_{-i})$ and     
$\te s_{-i} \in S_{-i} : p_i(s_i, s_{-i}) > p_i(s'_i, s_{-i})$}.                
\]
% \[                                                                            
% \fa s_{-i} \in S_{-i} \ p_i(s_i, s_{-i}) \geq p_i(s'_i, s_{-i}) \A            
% \]                                                                            
% and                                                                           
% \[                                                                            
% \te s_{-i} \in S_{-i} \ p_i(s_i, s_{-i}) > p_i(s'_i, s_{-i})                  
% \]                                                                            

Denote by $H^1$ a subgame of $H$ obtained by the elimination of some
(not necessarily all) strategies that are weakly dominated in $H$, and
put $H^{0} := H$ and $H^{k+1} := (H^k)^1$, where $k \geq 1$.
Note that $H^k$ is not uniquely defined, since we do not stipulate which
strategies are removed at each stage.

Abbreviate the phrase `iterated elimination of weakly dominated
strategies' to IEWDS. We say that each  $H^k$ is obtained from $H$
\bfe{by an IEWDS}.
If for some $k$, some subgame $H^k$ is a trivial
game we say that $H$ \bfe{can be solved by an IEWDS}. The relevant
result (that we shall not prove) is that if a finite strategic game $H$
can be solved by an IEWDS then every remaining joint strategy is a Nash
equilibrium of $H$. We shall illustrate it in Example
\ref{exa:centipede} in Section \ref{prelim-ext}.

The following lemma will be needed in Section \ref{sec:wd}.

\begin{lemma} \label{lem:wd}

  Let $H := (S_1, \LL, S_n, p_1, \LL, p_n)$ be a finite
  strategic game and let $H^k := (S^k_1, \LL, S^k_n, p_1, \LL, p_n)$,
  where $k \geq 1$.  Then
\[
    \fa i \in \{1, \LL, n\} \ \fa s_i \in S_i \ \te t_i \in S^k_i \ \fa s_{-i} \in S^k_{-i}:
p_i(t_i,s_{-i}) \geq  p_i(s_i,s_{-i}).
\]
\end{lemma}

\begin{proof}
  We proceed by induction. Take some $i \in \{1, \LL, n\}$ and
  $s_i \in S_i$.  Suppose $k=1$. If $s_i \in S^1_i$, then we are done, so
 assume that $s_i \not\in S^1_i$. $H$ is
  finite and the relation `weakly dominates' is transitive, so some
  strategy $t_i$ from $H$ weakly dominates $s_i$ in $H$ and is not
  weakly dominated in $H$, and thus is in $H^1$.

  Suppose the claim holds for some $k>1$.  By the induction hypothesis
  for some $u_i \in S^k_i$ we have
  $p_i(u_i,s_{-i}) \geq p_i(s_i,s_{-i})$ for all
  $s_{-i} \in S^k_{-i}$. 
  If $u_i \not\in S^{k+1}_i$, then for the same reasons as above some
  strategy $t_i$ from $H^{k+1}$ weakly dominates $u_i$ in $H^k$ and
  consequently $p_i(t_i,s_{-i}) \geq p_i(s_i,s_{-i})$ for all
  $s_{-i} \in S^k_{-i}$.
\end{proof}
\II

Finally, we introduce the following condition defined in \cite{MS97}.
We say that a strategic game   $(S_1, \LL, S_n, p_1, \LL, p_n)$
satisfies the \bfe{transference of decisionmaker indifference (TDI)}
condition if: 
\[
  \begin{array}{l}
\mbox{$\fa i \in \{1, \ldots, n\} \: \fa r_i, t_i \in S_i \: \fa s_{-i} \in S_{-i}$}: \\
\mbox{$p_{i}(r_i, s_{-i}) = p_{i}(t_i, s_{-i}) \to
p(r_i, s_{-i}) = p(t_i, s_{-i})$.}
  \end{array}
\]
Informally, this condition states that whenever for some player $i$ two of his strategies
$r_i$ and $t_i$ are indifferent w.r.t.~some joint strategy $s_{-i}$ of the other players
then this indifference extends to all players.

In the next section we shall introduce a natural class of strategic
games that satisfy the TDI condition.

\section{Preliminaries on strictly competitive games}
\label{sec:aux}

Sections \ref{sec:sc} and \ref{sec:comp} concern specific
extensive games that involve two players. It is convenient to introduce them
first as strategic games.  A strategic two-player game is called
\bfe{strictly competitive} if
\[
\fa i \in \{1,2\} \ \fa s,t \in S :
p_{i}(s) \geq p_{i}(t) \text{ iff } p_{-i}(s) \leq p_{-i}(t).
\]
(As there are here just two players, $-i$ denotes the opponent of
player $i$, so $p_{-i}(s)$ and $p_{-i}(t)$ are here numbers.)  Note
that every strictly competitive game satisfies the TDI condition as the
definition implies that
 \begin{equation}
\mbox{$\fa i \in \{1,2\} \ \fa s,t \in S : p_{i}(s) = p_{i}(t) \text{ iff
} p_{-i}(s) = p_{-i}(t)$.}
   \label{equ:iff}   
 \end{equation}

A two-player game is called \bfe{zero-sum} if 
\[
\fa s \in S :  p_1(s) + p_2(s) = 0.
\]
An example is the Matching Pennies game. Clearly every zero-sum game
is strictly competitive.

In Section \ref{sec:comp} we shall discuss two classes of zero-sum games.
A zero-sum game is called a \bfe{win or lose game} if the only
possible outcomes are $(1,-1)$ and $(-1,1)$, with 1 interpreted as a
\emph{win} and $-1$ as \emph{losing}. Finally, a zero-sum game is
called a \bfe{chess-like game} if the only possible outcomes are
$(1,-1)$, $(0, 0)$, and $(-1,1)$, with $0$ interpreted as a
\emph{draw}.

The following results about strictly competitive games will be needed
in Section \ref{sec:sc}.

\begin{lemma} \label{lem:sc-sum1}
 Consider a strictly competitive strategic game $H$ with a Nash
  equilibrium $s$.  Suppose that for some $i \in \{1,2\}$, $t_i$
  weakly dominates $s_i$.  Then $(t_i, s_{-i})$ is also a Nash
  equilibrium.
\end{lemma}

\begin{proof}
  Let $H := (S_1, S_2, p_1, p_2)$.
  Take a strategy $s'_i$ of player $i$. By the assumptions about $s$ and
  $t_i$
  \[
    p_i(t_i, s_{-i}) = p_i(s_i, s_{-i}) \geq p_i(s'_i, s_{-i}).
  \]

  Next, take a strategy $s'_{-i}$ of player $-i$.  By (\ref{equ:iff})
  and the fact that $s$ is a Nash equilibrium
  \[
    p_{-i}(t_i,s_{-i}) = p_{-i}(s_i,s_{-i}) \geq p_{-i}(s_i, s'_{-i}).
  \]
This establishes the claim.
\end{proof}
\II

Given a finite strategic game $H := (S_1, \LL, S_n, p_1, \LL, p_n)$
we define for each player $i$ 
\[
  maxmin_{i}(H) := \max_{s_{i} \in S_i} \min_{s_{-i} \in S_{-i}} p_i(s_i, s_{-i}).
\]
We call any strategy $s^{*}_i$ such that
$\min_{s_{-i} \in S_{-i}} p_i(s^*_i, s_{-i}) = maxmin_{i}(H)$ a
\bfe{security strategy} for player $i$ in $H$.

The following result goes back to \cite{vNM04}, where it was
established for zero-sum games. The formulation below is from
\cite[pages 22-23]{OR94}).

\begin{theorem}[Minimax]
  \label{thm:scminimax}
  Suppose that $s$ is a Nash equilibrium of a strictly competitive
  strategic game $H$. Then each $s_i$ is a security strategy for player $i$
  and $p(s) = (maxmin_{1}(H), maxmin_{2}(H))$.
\end{theorem}

\begin{corollary} \label{cor:h1} Consider a finite strictly
  competitive strategic game $H$ that has a Nash equilibrium. Then $H^1$
  has a Nash equilibrium, as well, and for all $i \in \{1,2\}$,
  $maxmin_{i}(H) = maxmin_{i}(H^1)$.
\end{corollary}
(The notation $H^1$ was introduced in Section \ref{sec:prelim}.)

\begin{proof}

We first prove that some Nash equilibrium of $H$ is also a joint strategy of $H^1$.
Let $s$ be a Nash equilibrium of $H$.  Suppose first that only one
strategy from $s$, say $s_i$, is not a strategy in $H^1$.  
The game $H$ is finite and the relation `weakly dominates'
is transitive so some strategy $t_i$ weakly dominates $s_i$ and is not weakly
dominated. Thus $(t_i, s_{-i})$ is a joint strategy in
$H^1$, which by Lemma \ref{lem:sc-sum1} is a Nash equilibrium in $H$.

Suppose now that none of the strategies from $s$ are strategies in
$H^1$.  By the argument just made we conclude that for some joint
strategy $t$ in $H^1$ first $(t_i, s_{-i})$ is a Nash equilibrium in $H$
and then that $t$ is a Nash equilibrium in $H$.

We conclude that a joint strategy is both a Nash equilibrium in $H$
and in $H^1$. The other claim then follows 
by the Minimax Theorem \ref{thm:scminimax}.
\end{proof}

\begin{lemma} \label{lem:scTwo}
  Consider a strictly competitive strategic game $H$ that has a Nash
  equilibrium and has two outcomes. Let $H^1$ be the result of removing
  from $H$ all weakly dominated strategies. Then $H^1$ is a trivial game.
\end{lemma}

\begin{proof}
  Let $s^*$ be a Nash equilibrium of $H = (S_1, S_2, p_1, p_2)$ and
  $s'$ a joint strategy such that $p(s^*)$ and $p(s')$ are the two
  outcomes in $H$. By condition (\ref{equ:iff}) from Section
  \ref{sec:prelim} $p_1(s^*) \neq p_1(s')$ and
  $p_2(s^*) \neq p_2(s')$.  $H$ is strictly competitive, so for some
  $i$ both $p_i(s^*) > p_i(s')$ and $ p_{-i}(s') > p_{-i}(s^*)$.

  First we show that $p_i(s^*_i, s_{-i}) = p_i(s^*)$ for all
  $s_{-i} \in S_{-i}$. Suppose otherwise. Take $s_{-i}$ such that
  $p_i(s^*_i, s_{-i}) \neq p_i(s^*)$.  Then
  $p_i(s^*_i, s_{-i}) = p_i(s')$, so by (\ref{equ:iff})
  $ p_{-i}(s^*_i, s_{-i}) = p_{-i}(s') > p_{-i}(s^*)$, which
  contradicts the fact that $s^*$ is a Nash equilibrium.
  
Hence by the choice of $i$ for all $s_{-i} \in S_{-i}$
\[
p_i(s^*_i, s_{-i}) = p_i(s^*) \geq p_i(s'_i, s_{-i})
\]
and
\[
  p_i(s^*_i, s'_{-i}) = p_i(s^*) > p_i(s').
\]
So $s^*_i$ weakly dominates $s'_i$. This implies that $H^1$ is a trivial game.
\end{proof}

\section{Extensive games}
\label{prelim-ext}

After these preliminaries we now focus on the subject of this tutorial.

A \bfe{rooted tree} (from now on, just a \bfe{tree}) is a connected
directed graph (i.e., such that the undirected version is connected) 
with a unique node with the in-degree 0, called the \bfe{root}, and in
which every other node has the in-degree 1.  A \bfe{leaf} is a node
with the out-degree 0. We denote a tree by $(V,E,v_0)$, where $V$ is a
non-empty set of nodes, $E$ is a possibly empty set of directed edges,
and $v_0$ is the root. In drawings the edges will be directed
downwards.

An \bfe{extensive game with perfect information} (in short, just
an \bfe{extensive game})  for $n \geq 1$ players consists of:

\begin{itemize}

\item a set of players $\{1, \LL, n\}$,

\item a \bfe{game tree} $T := (V,E,v_0)$; we denote its set of leaves by $Z$, 
  
\item a \bfe{turn function} $\turn: V \setminus Z \to \{1, \LL, n\}$,

\item the \bfe{outcome functions}
$o_i: Z \myra \mathbb{R}$, for each player $i$.
\end{itemize}
We denote it by $(T, \turn, o_1, \LL, o_n)$.

As in the case of strategic games we assume that each player is
rational (which now means that his objective is to maximize his
outcome in the game) and that the players have common knowledge of the
game and of each others' rationality.

A node $w$ is called a \bfe{child} of $v$ in $T$ if $(v,w) \in E$.  A node in
$T$ is called a \bfe{preleaf} if all its children are leaves.  We say
that an extensive game is \bfe{finite} if its game tree is finite. In
what follows we limit our attention to finite extensive games.

The function $\turn$ determines at each non-leaf node which player
should move.  The edges of $T$ represent possible \bfe{moves} in the
considered game, while for a node $v \in V \setminus Z$ the set of its
children $C(v) := \{w \mid (v,w) \in E\}$ represents possible
\bfe{actions} of player $\turn(v)$ at $v$.

In the figures below we identify the actions with the labels we put
on the edges and thus identify each action with the corresponding
move. For convenience we do not assume the labels to be unique, but it
will not lead to confusion.  Further, we annotate the non-leaf nodes
with the identity of the player whose turn it is to move and the name
of the node. Finally, we annotate each leaf node with the
corresponding sequence of the values of the $o_i$ functions.

\begin{example}\label{exa:1}

  Consider the Prisoner's Dilemma and Matching Pennies games from
  Example \ref{exa:stra}.  Suppose the players move sequentially with
  the row player moving first. The game trees of the resulting
  extensive games are depicted in Figures \ref{fig:prisoner}
  and \ref{fig:matching} below. The thick lines in the second drawing
  will be explained later.
  \HB
    \begin{figure}[htbp]
    \centering
 
% First, set the overall layout of the tree
% You might need to play with these sizes to ensure nothing overlaps.
\tikzstyle{level 1}=[level distance=1.5cm, sibling distance=5cm]
\tikzstyle{level 2}=[level distance=1.5cm, sibling distance=2.5cm]
\begin{tikzpicture}
%Start with the parent node, and slowly build out the tree
% with each "child" representing a new level of the diagram
% each "node" represents a labelled (or unlabeled if you
% want) node in the diagram.

\node {1,$u$}
child{
node{2,$v$}
child{
node{(2,2)}
edge from parent
node[left]{\scriptsize $C$}
}
child{
node{(0,3)}
edge from parent
node[right]{\scriptsize $D$}}
edge from parent
node[left]{\scriptsize $C$}
}
child{
node{2,$w$}
child{
node{(3,0)}
edge from parent
node[left]{\scriptsize $C$}
}
child{
node{(1,1)}
edge from parent
node[right]{\scriptsize $D$}}
edge from parent
node[right]{\scriptsize $D$}
};

\end{tikzpicture}
    \caption{Extensive form of the Prisoner's Dilemma game}
    \label{fig:prisoner}
  \end{figure}
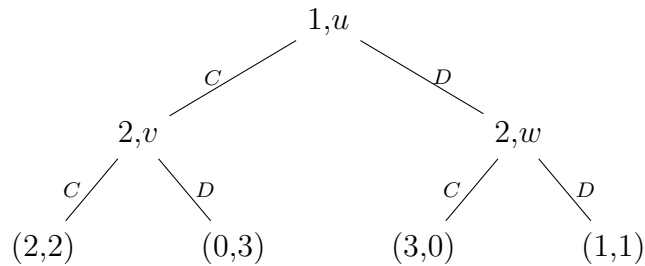

  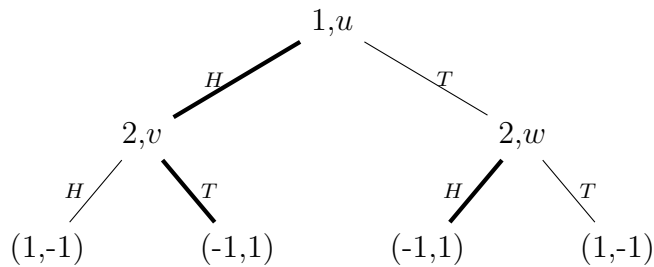
\begin{figure}[htbp]
    \centering
% First, set the overall layout of the tree
% You might need to play with these sizes to ensure nothing overlaps.
\tikzstyle{level 1}=[level distance=1.5cm, sibling distance=5cm]
\tikzstyle{level 2}=[level distance=1.5cm, sibling distance=2.5cm]
\tikzstyle{level 3}=[level distance=1.5cm, sibling distance=2cm]
\tikzstyle{level 4}=[level distance=1.5cm, sibling distance=2cm]
\tikzstyle{level 5}=[level distance=1.5cm, sibling distance=1cm]

\begin{tikzpicture}
[emph/.style={edge from parent/.style={very thick,draw}}]
[norm/.style={edge from parent/.style={black,thin,draw}}]

%Start with the parent node, and slowly build out the tree
% with each "child" representing a new level of the diagram
% each "node" represents a labelled (or unlabeled if you
% want) node in the diagram.

\node {1,$u$}
child{
node(e){2,$v$}
child{
node(f){(1,-1)}
edge from parent [thin]  
node[left]{\scriptsize $H$}
}
child{
node{(-1,1)}
edge from parent [ultra thick] 
node[right]{\scriptsize $T$}}
edge from parent [ultra thick] 
node[left]{\scriptsize $H$}
}
child{
node(a){2,$w$}
child{
node(b){(-1,1)}
edge from parent [ultra thick] 
node[left]{\scriptsize $H$}
}
child{
node{(1,-1)}
edge from parent [thin] 
node[right]{\scriptsize $T$}}
edge from parent [thin] 
node[right]{\scriptsize $T$}
};

\end{tikzpicture}
    \caption{Extensive form of the Matching Pennies game}
    \label{fig:matching}
  \end{figure}

%   \begin{figure}[htbp]
%     \centering
% % First, set the overall layout of the tree
% % You might need to play with these sizes to ensure nothing overlaps.
% \tikzstyle{level 1}=[level distance=1.5cm, sibling distance=5cm]
% \tikzstyle{level 2}=[level distance=1.5cm, sibling distance=2.5cm]
% \tikzstyle{level 3}=[level distance=1.5cm, sibling distance=2cm]
% \tikzstyle{level 4}=[level distance=1.5cm, sibling distance=2cm]
% \tikzstyle{level 5}=[level distance=1.5cm, sibling distance=1cm]
% \begin{tikzpicture}
% %Start with the parent node, and slowly build out the tree
% % with each "child" representing a new level of the diagram
% % each "node" represents a labelled (or unlabeled if you
% % want) node in the diagram.

% \node {1,$u$}
% child{
% node(e){2,$v$}
% child{
% node(f){(1,-1)}
% edge from parent
% node[left]{\scriptsize $H$}
% }
% child{
% node{(-1,1)}
% edge from parent
% node[right]{\scriptsize $T$}}
% edge from parent
% node[left]{\scriptsize $H$}
% }
% child{
% node(a){2,$w$}
% child{
% node(b){(-1,1)}
% edge from parent
% node[left]{\scriptsize $H$}
% }
% child{
% node{(1,-1)}
% edge from parent
% node[right]{\scriptsize $T$}}
% edge from parent
% node[right]{\scriptsize $T$}
% };

% \end{tikzpicture}
%     \caption{Extensive form of the Matching Pennies game}
%     \label{fig:matching}
%   \end{figure}

\end{example}

\begin{example} \label{exa:ultimatum}
The following two-player game is called the \bfe{Ultimatum game}.
Player 1 moves first and selects a number $x \in \{0, 1, \LL, 100\}$
intepreted as a percentage of some good to be shared, leaving the
fraction of $(100-x) \%$ for the other player. Player 2 either
\emph{accepts} this decision, the outcome is then $(x, 100-x)$, or
\emph{rejects} it, the outcome is then $(0,0)$.  The game tree is
depicted in Figure \ref{fig:ultimatum}, where the action of player 1
is a number from the set $\{0, 1, \LL, 100\}$, and the actions of
player 2 are $A$ and $R$.

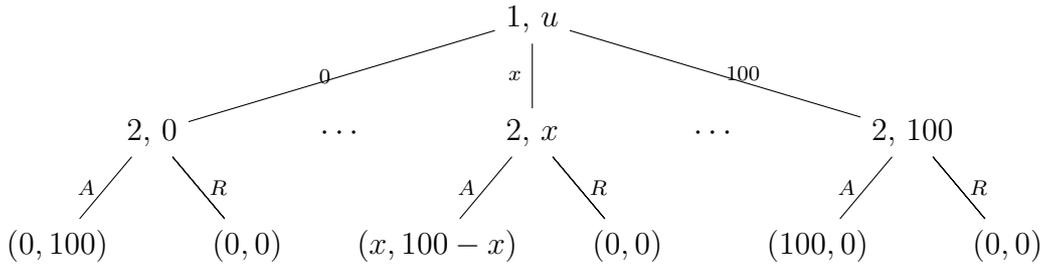
\begin{figure}[ht]
  \centering
  \tikzstyle{level 1}=[level distance=1.5cm, sibling distance=5cm]
  \tikzstyle{level 2}=[level distance=1.5cm, sibling distance=2.5cm]
  \tikzstyle{level 3}=[level distance=1.5cm, sibling distance=2cm]
\begin{tikzpicture}
 \node (r){1, $u$}
 child{
   node (a){2, $0$}
   child{
     node (d){$(0, 100)$}
     edge from parent
     node[left]{\scriptsize $A$}
   }
   child{
     node(e){$(0,0)$}
     edge from parent
     node[right]{\scriptsize $R$}
     edge from parent
   }
   edge from parent
   node[left]{\scriptsize $0$}
   }
 child{
   node (b){2, $x$}
   child{
     node (f){$(x, 100-x)$}
     edge from parent
     node[left]{\scriptsize $A$}
   }
   child{
     node (g){$(0,0)$}
     edge from parent
     node[right]{\scriptsize $R$}
     edge from parent
   }
   edge from parent
   node[left]{\scriptsize $x$}
 }
  child{
   node (c){2, 100}
   child{
     node (h){$(100, 0)$}
     edge from parent
     node[left]{\scriptsize $A$}
   }
   child{
     node (i){$(0,0)$}
     edge from parent
     node[right]{\scriptsize $R$}
     edge from parent
   }
  edge from parent
  node[right]{\scriptsize $100$}};

  \path (a) -- (b) node [midway] {$\cdots$};
  \path (b) -- (c) node [midway] {$\cdots$};
\end{tikzpicture}
    \caption{The Ultimatum game}
    \label{fig:ultimatum}
  \end{figure}
\HB  
\end{example}

Next we introduce strategies in extensive games.  Consider a finite
extensive game $G:= (T, \turn, o_1, \LL, o_n)$.  Let $V_i := \{v \in V
\setminus Z \mid \turn(v) = i\}$.  So $V_i$ is the set of nodes at which
player $i$ moves.  Its elements are called the \bfe{decision nodes} of
player $i$.  A \bfe{strategy} for player $i$ is a function $s_i: V_i
\to V$, such that $(v, s_i(v)) \in E$ for all $v \in V_i$.  \bfe{Joint
  strategies} are defined as in strategic games.  When the game tree
consists of just one node, each strategy is the empty function,
denoted by $\emptyset$, and there is only one joint strategy, namely
the $n$-tuple of these functions.

Each joint strategy different from $(\ES, \LL, \ES)$ assigns a unique child to
every node in $V \setminus Z$.  In fact, we can identify joint
strategies with such assignments.  Each joint strategy $s = (s_1, \LL,
s_n)$ determines a rooted path $\mathit{play}(s) := (v_1, \LL, v_m)$ in
$T$ defined inductively as follows:

\begin{itemize}
\item $v_1$ is the root of $T$,

\item if $v_{k} \not\in Z$, then $v_{k+1} := s_i(v_k)$, where $\turn(v_k) = i$.

\end{itemize}

Informally, given a joint strategy $s$, we can view $\mathit{play}(s)$ as the
resulting \emph{play} of the game.

$G$ is finite, so for each joint
strategy $s$ the rooted path $\mathit{play}(s)$ is finite.  Denote by
$\leaf(s)$ the last element of $\mathit{play}(s)$. We call then
$(o_1(\leaf(s)), \LL, o_n(\leaf(s)))$ the \bfe{outcome} of the game
$G$ when each player $i$ pursues his strategy $s_i$ and abbreviate it
to $o(\leaf(s))$.

\begin{example}
  \label{exa:extMP}
  
  Let us return to the extensive form of the Matching Pennies game
  from Example \ref{exa:1}.  The strategies for player 1 are: $H$ and  $T$,
  while the strategies for player 2 are: $\mathit{HH}$, $\mathit{HT}$,
  $\mathit{TH}$, and $\mathit{TT}$, where for instance $\mathit{TH}$
  stands for a strategy that selects $T$ at the node $v$ and $H$ at
  the node $w$.  In Figure \ref{fig:matching} thick lines correspond
  with the joint strategy $(H,\mathit{TH})$. Here
  $\play(\mathit{H,TH}) = (u, v, (-1,1))$, where we identify each leaf
  with the corresponding outcome, and $o(\leaf(\mathit{H,TH})) = (-1,1)$. 
  \HB
\end{example}

With each finite extensive game $G: = (T, \turn, o_1, \LL, o_n)$ we
associate a strategic game
$\Gamma(G):= (S_1, \LL, S_n , p_1, \LL, p_n)$ defined as follows:

\begin{itemize}

\item $S_i$ is the set of strategies of player $i$ in $G$,

\item $p_i(s) := o_i(\leaf(s))$.

\end{itemize}
We call $\Gamma(G)$ the \bfe{strategic form} of $G$.

All notions introduced in the context of strategic games can now be reused
in the context of finite extensive games simply by referring to the
corresponding strategic form. This way we obtain the notions of a best
response, Nash equilibrium, extensive games that have the FIP,
are weakly acyclic, etc.

\begin{example}
  The strategic form of the extensive form of the Matching Pennies game from Example
  \ref{exa:1} differs from the initial Matching Pennies game from Example
  \ref{exa:stra} and looks as follows:

\begin{center}
\begin{game}{2}{4}
      & $HH$               & $HT$                & $TH$               & $TT$ \\
$H$   &$\phantom{-}1,-1$    &$\phantom{-}1,-1$   &{$-1,\phantom{-}1$}    &$-1, \phantom{-}1$ \\
$T$   &$-1, \phantom{-}1$   &$\phantom{-}1, -1$  &{$-1,\phantom{-}1$}   &$\phantom{-}1,-1$   
\end{game}
\end{center}
\II

Note that this game has two Nash equilibria: $(H, \mathit{TH})$ and
$(T, \mathit{TH})$. The first one is depicted in Figure
\ref{fig:matching} by thick lines. By definition, these are the Nash
equilibria of the extensive form of the Matching Pennies game.

The intuitive reason that there are two Nash equilibria is that no
matter which out of his two strategies the first player selects, the
second player can always secure the payoff 1 for himself. Of course,
the decision which player moves first affects both the sets of
strategies and the sets of Nash equilibria.

One can also easily check that the extensive form of the Prisoner's Dilemma
game given in Figure \ref{fig:prisoner} has one Nash equilibrium,
$(D, DD)$.
\HB
\end{example}

\begin{example} \label{exa:ultimatum1}
  Consider now the Ultimatum game from Example \ref{exa:ultimatum}.
  Each strategy for player 1 is a number from $\{0, 1, \LL, 100\}$,
  while each strategy for player 2 is a function from $\{0, 1, \LL, 100\}$
  to $\{A, R\}$.
    % So it should be clear what we mean by player 2
   % strategies \textbf{always} $A$ or \textbf{if} $x \leq x^*$
   % \textbf{then} $A$ \textbf{else} $R$ \textbf{fi}, where
   % $x^* \in [0,100]$ or in the discrete variant
   % $x^* \in \{0, 1, \LL, 100\}$.

  It is easy to check that $(100, s_2)$, where for
  $y \in \{0, 1, \LL, 100\}$, $s_2(y) = R$ is a Nash equilibrium with
  the outcome $(0,0)$ and that all other Nash equilibria are of the
  form $(x,s_2)$, where $s_2(x) = A$ and $s_2(y) = R$ for $y > x$, and
  $x,y \in \{0, 1, \LL, 100\}$, with the corresponding outcome
  $(x, 100-x)$.
   \HB
\end{example}

Concepts such as Nash equilibrium can be defined directly, without a
detour through the strategic games.  However, introducing strategic
games first allows us to view finite extensive games as a special
class of strategic games and allows us to conclude that some results
established for the strategic games, for instance the Weakly Acyclic
Theorem \ref{thm:weakly}, also hold for all finite extensive games.

As we shall see, for extensive games, due to their structure,
additional results hold. Further, their structure suggests a new
equilibrium notion that is meaningful only for these games.  Finally,
when discussing iterated elimination of weakly dominated strategies in
an extensive game, one needs to reason about its strategic form.

All examples in this section are instances of the so-called
\bfe{Stackelberg competition}.  In such games a \emph{leader} moves
first and a \emph{follower}, having observed the resulting action,
moves second.  Of course, there are other natural extensive games, in
particular multi-player games and games with infinite game trees.

For extensive games that are not Stackelberg competition games it is
legitimate to question the notion of a strategy. Namely, one would
expect that when a player following a strategy makes a move to a node
$u$, then all his subsequent moves should take place in the subtree
rooted at $u$. However, the definition of a strategy does not stipulate it
as it is `overdefined'.  A natural revision was
introduced in \cite[page 94]{OR94}.

Given a node $w$ in the game tree consider the path 
$v = v_1, \LL, v_k = w$ to it from the root $v$. Let
\[
  [w]_i := \{(v_j, v_{j+1}) \mid j \in \{1, \LL, k-1 \} \mathrm{\ and \ } \turn(v_j) = i\}.
\]
Informally, $[w]_i$ is the set of the moves of player $i$ that lie on
the path from the root to $w$.

We call $rs_i$ a \bfe{reduced strategy} of player $i$ if it is a
maximal subset of a strategy $s_i$ (recall that each strategy is a
function, so a set of pairs of nodes) that satisfies the following
property:
\[
\mbox{if $(u,w) \in rs_i$, then $[w]_i \sse rs_i$}.
\]
Intuitively, $rs_i$ is a reduced strategy of player $i$ if according to it 
each of his moves follows his earlier moves in the game.

Just like the joint strategies, each joint reduced strategy determines a play of the game. This allows us
to define the outcome of the game when each player pursues his reduced strategy.
Associate now with each joint reduced strategy $r$ the following set of original joint strategies:
\[
  \emph{Str}(r) := \{s \mid \fa i \in \{1, \LL n\}: r_i \sse s_i\}.
\]
One can show that the sets $\emph{Str}(r)$, where $r$ is a Nash
equilibrium in the reduced strategies, form a partition of the set of
original Nash equilibria.  So each Nash equilibrium in the reduced
strategies is a convenient representation of a set of Nash equilibria
in the original strategies.

\begin{example} \label{exa:centipede}
  Consider the classic centipede game due to \cite{Ros81}.  In Figure
  \ref{fig:centipede} we present a version of the game from \cite[page
  107]{OR94}. $C$ and $S$ represent the actions `continue' and `stop'.

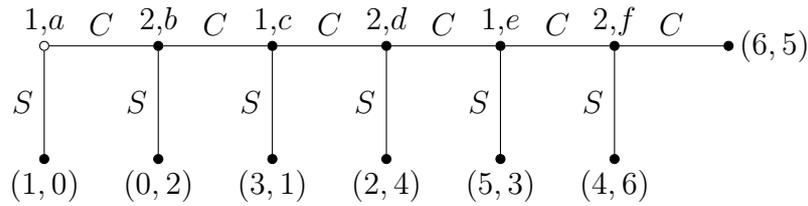
\begin{figure}[htbp]
\centering

\begin{tikzpicture}
% Two node styles: solid and hollow
\tikzstyle{solid node}=[circle,draw,inner sep=1.2,fill=black];
\tikzstyle{hollow node}=[circle,draw,inner sep=1.2];
% The Tree
\node(0)[hollow node]{}
child[grow=down]{node[solid node]{}edge from parent node[left]{$S$}}
child[grow=right]{node(1)[solid node]{}
child[grow=down]{node[solid node]{}edge from parent node[left]{$S$}}
child[grow=right]{node(2)[solid node]{}
child[grow=down]{node[solid node]{}edge from parent node[left]{$S$}}
child[grow=right]{node(3)[solid node]{}
child[grow=down]{node[solid node]{}edge from parent node[left]{$S$}}
child[grow=right]{node(4)[solid node]{}
child[grow=down]{node[solid node]{}edge from parent node[left]{$S$}}
child[grow=right]{node(5)[solid node]{}
child[grow=down]{node[solid node]{}edge from parent node[left]{$S$}}
child[grow=right]{node(6)[solid node]{}
edge from parent node[above]{$C$}
}
edge from parent node[above]{$C$}
}
edge from parent node[above]{$C$}
}
edge from parent node[above]{$C$}
}
edge from parent node[above]{$C$}
}
edge from parent node[above]{$C$}
};
% Movers
\node[above]at(0){1,$a$};
\node[above]at(2){1,$c$};
\node[above]at(4){1,$e$};
\node[above]at(1){2,$b$};
\node[above]at(3){2,$d$};
\node[above]at(5){2,$f$};

%\foreach \x in {1,3,5}
%\node[above]at(\x){2};

% payoffs
\node[below]at(0-1){$(1,0)$};
\node[below]at(1-1){$(0,2)$};
\node[below]at(2-1){$(3,1)$};
\node[below]at(3-1){$(2,4)$};
\node[below]at(4-1){$(5,3)$};
\node[below]at(5-1){$(4,6)$};
\node[right]at(6){$(6,5)$};
\end{tikzpicture}
      \caption{A six-period version of the centipede game}
      \label{fig:centipede}
    \end{figure}

  Each player has three decision nodes and two actions at each of
  them. So each player has eight strategies. In contrast, both players
  have four reduced strategies. These are
  \II

  for player 1: $aS$, $aC\s{}cS$, $aC\s{}cC\s{}eS$, $aC\s{}cC\s{}eC$, and
  \II
  
  for player 2: $bS$, $bC\s{}dS$,  $bC\s{}dC\s{}fS$, $bC\s{}dC\s{}fC$,
  \II
  
  \NI
  where $aC\s{}cC\s{}eS$, is a shorthand for $\{(a,C), (c,C), (e,S)\}$,
  etc.  The strategic form corresponding to the joint reduced strategies
  is given in Figure \ref{fig:strategic}.

\begin{figure}[htbp]
  \centering
\begin{game}{4}{4}
                 & $bS$    & $bC\s{}dS$    & $bC\s{}dC\s{}fS$ & $bC\s{}dC\s{}fC$ \\
$aS$             & $1,0$   &$1,0$        & $1,0$        & $1,0$ \\
$aC\s{} cS$        & $0,2$   &$3,1$        & $3,1$        & $3,1$ \\
$aC\s{} cC\s{} eS$   & $0,2$   &$2,4$        & $5,3$        & $5,3$ \\
$aC\s{} cC\s{} eC$   & $0,2$   &$2,4$        & $4,6$        & $6,5$
\end{game}

  \caption{The strategic form of the centipede game that uses reduced strategies}
  \label{fig:strategic}
\end{figure}

This game has a unique Nash equilibrium, namely $(aS, bS)$, with the
outcome $(1,0)$.  It can be obtained by solving the game by an IEWDS
in six steps, by repeatedly eliminating the rightmost column and the
lowest row.  In contrast, the strategic form corresponding to the 
original extensive game has several Nash equilibria.
According to the notation introduced above, these Nash equilibria
form the set $\emph{Str}((aS, bS))$.
\HB
\end{example}

\section{Subgame perfect equilibria}
\label{sec:spe}

\subsection{Definition and examples}

Example \ref{exa:ultimatum1} suggests that the concept of a Nash
equilibrium is not informative enough to predict outcomes of extensive
games: it results in too many scenarios, some of which are obviously
inferior to all players.  Another issue is the problem
of so-called `not credible threat' as illustrated in the following
example.

 \begin{figure}[htbp]
    \centering
\tikzstyle{level 1}=[level distance=1.5cm, sibling distance=5cm]
\tikzstyle{level 2}=[level distance=1.5cm, sibling distance=2.5cm]
\tikzstyle{level 3}=[level distance=1.5cm, sibling distance=2cm]
\tikzstyle{level 4}=[level distance=1.5cm, sibling distance=2cm]
\tikzstyle{level 5}=[level distance=1.5cm, sibling distance=1cm]

\begin{tikzpicture}
[emph/.style={edge from parent/.style={very thick,draw}}]
[norm/.style={edge from parent/.style={black,thin,draw}}]

\node {1,$u$}
child{
node(e){2,$v$}
child{
node(f){(1,-1)}
edge from parent [thin]  
node[left]{\scriptsize $H$}
}
child{
node{(-1,1)}
edge from parent [thin]
node[right]{\scriptsize $T$}}
edge from parent [thin]
node[left]{\scriptsize $H$}
}
child{
node(a){2,$w$}
child{
node(b){(-1,1)}
edge from parent [thin]
node[left]{\scriptsize $H$}
}
child{
node{(-10, 0)}
edge from parent [thin] 
node[right]{\scriptsize $T$}}
edge from parent [thin] 
node[right]{\scriptsize $T$}
};

\end{tikzpicture}
    \caption{A modification of the Matching Pennies game}
    \label{fig:matching1}
  \end{figure}
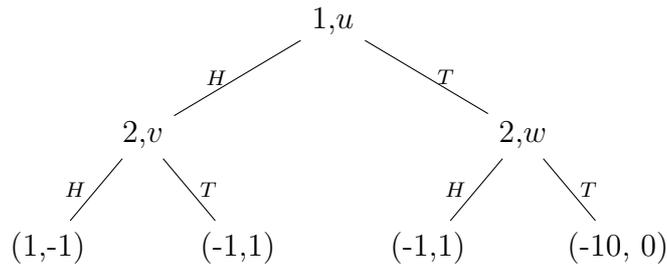

\begin{example}
  \label{ex:notcredible}
Consider the extensive game given in Figure \ref{fig:matching1}. This
game has three Nash equilibria: $(H, TH)$, $(T, TH)$, and $(H,
TT)$. However, the last equilibrium is not plausible: if player 1
chooses $T$, then player 2 should select $H$ and not $T$: the `threat'
of player 2 to choose $T$ at the node 2 is not credible.
\HB
\end{example}

Motivated by the issue of non-credible threats Selten introduced in \cite{Sel65} a
stronger equilibrium concept.  To define it we need to use 
strategies instead of the restricted strategies.

Consider an extensive game $G: = (T, \turn, o_1, \LL, o_n)$ and a non-leaf
node $w$ of $T$. Denote by $T^w$ the subtree of $T$ rooted at $w$.
The \bfe{subgame of $G$ rooted at the node $w$}, denoted by
$G^w$, is defined as follows:

\begin{itemize}

\item its set of players is $\{1, \LL, n\}$,
  
\item its tree is $T^w$,

\item its turn and payoff functions are the restrictions of 
the corresponding functions of $G$ to the nodes of $T^w$.

\end{itemize}
So the notion of a subgame has a different meaning for the strategic
and for the extensive games.  Note that some players may `drop out' in
$G^w$, in the sense that at no node of $T^w$ it is their turn to move.
Still, to keep the notation simple, it is convenient to admit in $G^w$
all original players in $G$. Each strategy $s_i$ of player $i$ in $G$
uniquely determines his strategy $s^w_i$ in $G^w$.  Given a joint
strategy $s = (s_1, \LL, s_n)$ of $G$ we denote by $s^w$ the joint
strategy $(s^w_1, \LL, s^w_n)$ in $G^w$.

A joint strategy $s$ of $G$ is called a \bfe{subgame perfect
  equilibrium} in $G$ if for every node $w$ of $T$, the joint strategy
$s^w$ of $G^w$ is a Nash equilibrium in $G^w$.  Informally $s$ is
subgame perfect equilibrium in $G$ if it induces a Nash equilibrium in
every subgame of $G$.

It is straightforward to check that all Nash equilibria in the
extensive forms of the Prisoner's Dilemma and Matching Pennies games
are also subgame perfect equilibria. The modified Matching
  Pennies game given in Example \ref{ex:notcredible} has two subgame
  perfect equilibria: $(H,\mathit{TH})$ and $(T,\mathit{TH})$. Note
  that the Nash equilibrium $(H, \mathit{TT})$ that involves a
  non-credible threat is not a subgame perfect equilibrium.

\begin{example} \label{exa:ultimatum2}

  Return now to the Ultimatum game from Example
  \ref{exa:ultimatum}. In Example \ref{exa:ultimatum1} we noticed that
  this game has several Nash equilibria. However, only two of them are
  subgame perfect equilibria. They are depicted in Figures
  \ref{fig:ultimatum1} and \ref{fig:ultimatum2}. In the first
  equilibrium player 1 selects 100 and player 2 accepts all offers,
  while in the second equilibrium player 1 selects 99 and player 2
  accepts all offers except 100. These equilibria are more intuitive
  than the remaining Nash equilibria and provide natural insights into
  this game.
  \HB
\end{example}

\begin{figure}[htbp]
  \centering
  \tikzstyle{level 1}=[level distance=1.5cm, sibling distance=5cm]
  \tikzstyle{level 2}=[level distance=1.5cm, sibling distance=2.5cm]
  \tikzstyle{level 3}=[level distance=1.5cm, sibling distance=2cm]
\begin{tikzpicture}
 \node (r){1, $u$}
 child{
   node (a){2, $0$}
   child{
     node (d){$(0, 100)$}
     edge from parent [ultra thick]
     node[left]{\scriptsize $A$}
   }
   child{
     node(e){$(0,0)$}
     edge from parent
     node[right]{\scriptsize $R$}
     edge from parent
   }
   edge from parent
   node[left]{\scriptsize $0$}
   }
 child{
   node (b){2, $x$}
   child{
     node (f){$(x, 100-x)$}
     edge from parent [ultra thick]
     node[left]{\scriptsize $A$}
   }
   child{
     node (g){$(0,0)$}
     edge from parent [thin]
     node[right]{\scriptsize $R$}
     edge from parent
   }
   edge from parent 
   node[left]{\scriptsize $x$}
 }
  child{
   node (c){2, 100}
   child{
     node (h){$(100, 0)$}
     edge from parent [ultra thick]
     node[left]{\scriptsize $A$}
   }
   child{
     node (i){$(0,0)$}
     edge from parent
     node[right]{\scriptsize $R$}
     edge from parent [thin]
   }
  edge from parent [ultra thick]
  node[right]{\scriptsize $100$}};

  \path (a) -- (b) node [midway] {$\cdots$};
  \path (b) -- (c) node [midway] {$\cdots$};
\end{tikzpicture}
    \caption{A subgame perfect equilibrium in the Ultimatum game}
    \label{fig:ultimatum1}
\end{figure}
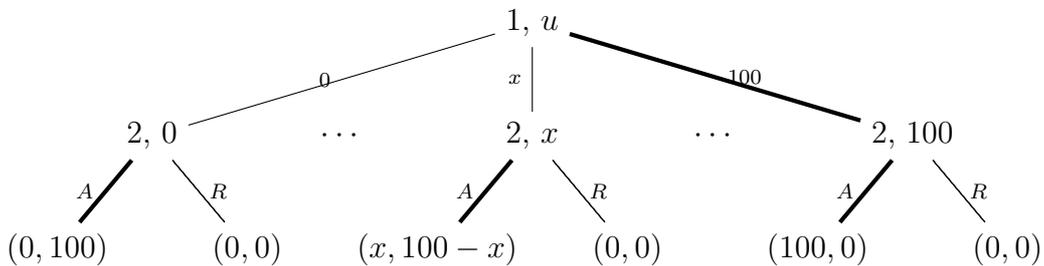

\begin{figure}[htbp]
  \centering
  \tikzstyle{level 1}=[level distance=1.5cm, sibling distance=3.5cm]
  \tikzstyle{level 2}=[level distance=1.5cm, sibling distance=2cm]
  \tikzstyle{level 3}=[level distance=1.5cm, sibling distance=0.1cm]
\begin{tikzpicture}
 \node (r){1, $u$}
 child{
   node (a){2, $0$}
   child{
     node (d){$(0, 100)$}
     edge from parent [ultra thick]
     node[left]{\scriptsize $A$}
   }
   child{
     node(e){$(0,0)$}
     edge from parent
     node[right]{\scriptsize $R$}
     edge from parent
   }
   edge from parent
   node[left]{\scriptsize $0$}
}
child{
  node (z){2, $98$}
   child{
     node (d){$(98, 2)$}
     edge from parent [ultra thick]
     node[left]{\scriptsize $A$}
   }
   child{
     node(e){$(0,0)$}
     edge from parent
     node[right]{\scriptsize $R$}
     edge from parent
   }
   edge from parent
   node[left]{\scriptsize $98$}
   }
 child{
   node (b){2, $99$}
   child{
     node (f){$(99, 1)$}
     edge from parent [ultra thick] 
     node[left]{\scriptsize $A$}
   }
   child{
     node (g){$(0,0)$}
     edge from parent [thin] 
     node[right]{\scriptsize $R$}
     edge from parent
   }
   edge from parent [ultra thick]
   node[left]{\scriptsize $99$}
 }
  child{
   node (c){2, 100}
   child{
     node (h){$(100, 0)$}
     edge from parent
     node[left]{\scriptsize $A$}
   }
   child{
     node (i){$(0,0)$}
     edge from parent
     node[right]{\scriptsize $R$}
     edge from parent [ultra thick]
   }
  edge from parent
  node[right]{\scriptsize $100$}};

  \path (a) -- (z) node [midway] {$\cdots$};
  % \path (b) -- (c) node [midway] {$\cdots$};
\end{tikzpicture}
    \caption{Another subgame perfect equilibrium in the Ultimatum game}
    \label{fig:ultimatum2}
\end{figure}
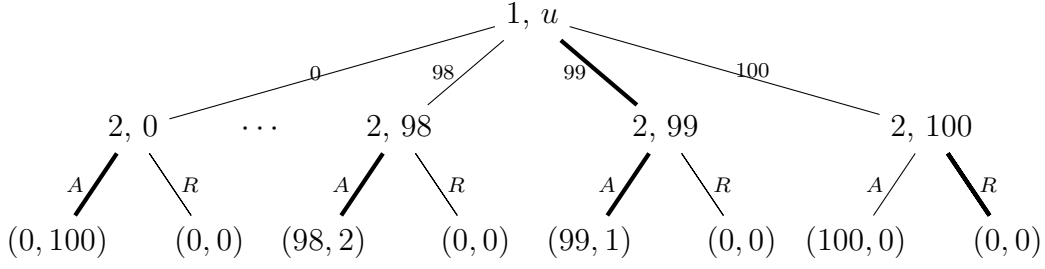

It should be noted that the Ultimatum game has been extensively
analysed in experimental economics. It has been observed that in
practice, people do not often play a Nash equilibrium or a subgame
perfect equilibrium.

\subsection{Backward induction}

We now show that for a finite extensive game $G$ over $T$, one can
construct a subgame perfect equilibrium using the iterative procedure
described in Algorithm \ref{alg:BI} called the \bfe{backward
  induction algorithm}. Since each loop iteration of Algorithm
\ref{alg:BI} modifies the underlying game tree, we use the notation
$C(v,T)$ to denote the set of children of node $v$ in the current
version of the tree $T$.

{\footnotesize
  \begin{algorithm}[htbp]
%    \caption{\label{alg:EBI}}
\caption{\label{alg:BI}}
\KwIn{A finite extensive game $G: = (T, \turn, o_1, \LL, o_n)$ with $T = (V,E,v_0)$} 
\KwOut{A subgame perfect equilibrium $s$ in $G$ and
  extensions of the functions $o_1, \LL, o_n$ to all nodes of $T$
  such that $o(v_0) = o(\leaf(s))$}
\While{$|V| > 1$}{
  \textbf{choose} $v \in V$ that is a preleaf of $T$;

  $i := \turn(v)$;

  \textbf{choose} $w \in C(v,T)$ such that $o_i(w)$ is maximal;

  $s_i(v) := w$;

  $o(v) := o(w)$;

  $V := V \setminus C(v,T)$;

  $E := E \cap (V \times V)$;

  $T := (V,E,v_0)$

}  
\end{algorithm}
}

Note that Algorithm \ref{alg:BI} always terminates but in general need
not have a unique outcome due to the presence of the \textbf{choose}
statements.  Each execution (i.e., each selection of values in the \textbf{choose} statements)
constructs a unique joint strategy $s$ and
an extension of the functions $o_1, \LL, o_n$ to all nodes of the game
tree.

\begin{example}
  Consider the Ultimatum game from Example
  \ref{exa:ultimatum}. Algorithm \ref{alg:BI} generates two possible
  outputs that correspond to Figures \ref{fig:ultimatum1} and
  \ref{fig:ultimatum2}.  For the second outcome the corresponding
  extensions of the outcome functions to all nodes are given in Figure
  \ref{fig:ultimatum3}.
  \HB

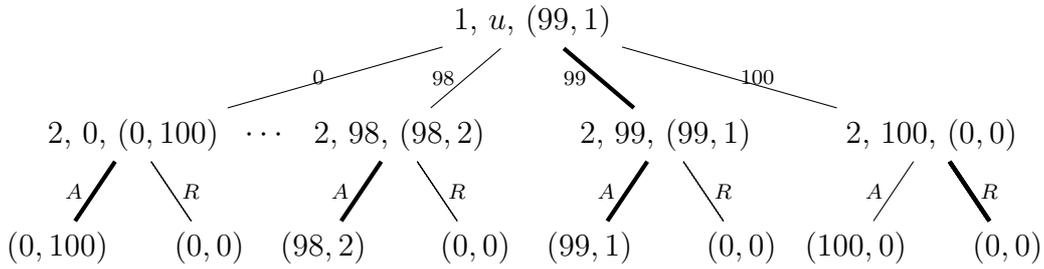
\begin{figure}[htbp]
  \centering
  \tikzstyle{level 1}=[level distance=1.5cm, sibling distance=3.5cm]
  \tikzstyle{level 2}=[level distance=1.5cm, sibling distance=2cm]
  \tikzstyle{level 3}=[level distance=1.5cm, sibling distance=0.1cm]
\begin{tikzpicture}
 \node (r){1, $u$, $(99,1)$}
 child{
   node (a){2, $0$, $(0,100)$}
   child{
     node (d){$(0, 100)$}
     edge from parent [ultra thick]
     node[left]{\scriptsize $A$}
   }
   child{
     node(e){$(0,0)$}
     edge from parent
     node[right]{\scriptsize $R$}
     edge from parent
   }
   edge from parent
   node[left]{\scriptsize $0$}
}
child{
  node (z){2, $98$, $(98,2)$}
   child{
     node (d){$(98, 2)$}
     edge from parent [ultra thick]
     node[left]{\scriptsize $A$}
   }
   child{
     node(e){$(0,0)$}
     edge from parent
     node[right]{\scriptsize $R$}
     edge from parent
   }
   edge from parent
   node[left]{\scriptsize $98$}
   }
 child{
   node (b){2, $99$, $(99,1)$}
   child{
     node (f){$(99, 1)$}
     edge from parent [ultra thick] 
     node[left]{\scriptsize $A$}
   }
   child{
     node (g){$(0,0)$}
     edge from parent [thin] 
     node[right]{\scriptsize $R$}
     edge from parent
   }
   edge from parent [ultra thick]
   node[left]{\scriptsize $99$}
 }
  child{
   node (c){2, 100, $(0,0)$}
   child{
     node (h){$(100, 0)$}
     edge from parent
     node[left]{\scriptsize $A$}
   }
   child{
     node (i){$(0,0)$}
     edge from parent
     node[right]{\scriptsize $R$}
     edge from parent [ultra thick]
   }
  edge from parent
  node[right]{\scriptsize $100$}};

  \path (a) -- (z) node [midway] {$\cdots$};
  % \path (b) -- (c) node [midway] {$\cdots$};
\end{tikzpicture}
    \caption{The backward induction algorithm and the Ultimatum game}
    \label{fig:ultimatum3}
\end{figure}
\end{example}

The following characterisation result makes use of the nondeterminism
present in the algorithm.

\begin{theorem}
  \label{thm:extSPE}
  For every finite extensive game all possible executions of the backward
  induction algorithm generate precisely all subgame perfect
  equilibria.

\end{theorem}

To establish this result we shall need a preparatory lemma, called the
`one deviation property' (see, e.g., \cite[page 98]{OR94}). Recall
that for a function $f: X \to Y$ (with $X \neq \emptyset$),
$\mathrm{argmax}_{x \in X} f(x) := \{y \in X \mid f(y) = \max_{x \in
  X} f(x)\}$.

\begin{lemma}
  \label{lm:SPE}
    Let $G$ be a finite extensive game over the game tree $T$. A joint
    strategy $s$ is a subgame perfect equilibrium in $G$ iff
for all non-leaf nodes $u$ in $T$
\[
\mbox{$s_i(u) \in \mathrm{argmax}_{x \in C(u)} o_i(\leaf(s^x))$, where $i=\turn(u)$.}
\]  
\end{lemma}
In words, this condition states that for all non-leaf nodes $u$ in $T$
and $i=\turn(u)$, $s_i(u)$ selects a child $x$ of $u$ for which
$o_i(\leaf(s^x))$ is maximal.  

For a proof see, e.g, \cite[pages 98-99]{OR94} or a more detailed
presentation in the appendix of \cite{AS21}.

\begin{corollary}[\cite{AS21}, Corollary 7]
  \label{lm:SPE2}
Let $G$ be a finite extensive game over the game tree $T$ with the
root $v$. A joint strategy $s$ is a subgame perfect equilibrium in $G$
iff for all $u \in C(v)$

  \begin{itemize}
\item $s_i(v) \in \mathrm{argmax}_{x \in C(v)} o_i(\leaf(s^x))$, where $i=\turn(v)$, 

  \item $s^{u}$ is a subgame perfect equilibrium in the subgame $G^{u}$.
  \end{itemize}
\end{corollary}

Intuitively, the first condition states that among the subgames rooted
at the children of the root $v$, the one determined by the
first move in the game $G$ yields the best outcome for the player who
moved.
\II

%% \begin{proof}
%% If $C(v)=\emptyset$, the claim is vacuously true. Otherwise consider
%% any $u \in C(v)$.  By Lemma \ref{lm:SPE} $s^u$ is a subgame perfect
%% equilibrium in $G^u$ iff for all non-leaf nodes $y$ in $T^u$ and $z
%% \in C(y)$, $ o_i(\leaf((s^{u})^{x}) \geq o_i(\leaf((s^{u})^{z})$, 
%% where $i=\turn(y)$ and $s^u_i(y)=x$.

%% Since $(s^{u})^{x}=s^{x}$ and $(s^{u})^{z}=s^{z}$, the last statement is
%% equivalent to the statement that the inequality in Lemma \ref{lm:SPE}
%% holds for all non-leaf nodes $y$ in $T^u$ and $z \in C(y)$.
%% %
%% The conclusion now follows by Lemma \ref{lm:SPE}.
%% \end{proof}
%% \II

\NI \emph{Proof of Theorem \ref{thm:extSPE}.}
%\begin{proof}
  %% We prove a stronger claim, namely that all these executions in
  %% addition to a subgame perfect equilibrium $s$ also generate the
  %% extensions of the functions $o_1, \LL, o_n$ to all nodes of $T$
  %% such that for the root $v_0$ of $T$ and all $j \in \{1, \LL, n\}$
  %% we have $o_j(v_0) = o_j(leaf(s))$.
%%  
  %% For a finite tree $T$ let $\tdepth(T)$ denote the number of edges
  %% on the longest path from the root to a leaf in the tree $T$.
  %% Below $T$ is the game tree of the considered extensive game $G$.
%% We argue by induction on $\tdepth(T)$.
%%
The proof proceeds by induction on $\mathit{height}(T)$, defined as
the number of edges in the longest path in the tree $T$.  The base
case when $\mathit{height}(T)=1$ is straightforward. Suppose
$\mathit{height}(T) > 1$.  Let $C(v) = \{w_1,\ldots w_k\}$.

\NI $(\Ra)$ Consider a joint strategy $s$ in $G$ together with some
extensions of the functions $o_1, \LL, o_n$ to the nodes of the game
tree of $G$ that are generated by an execution of Algorithm \ref{alg:BI}.

Fix an arbitrary $l \in \{1,\ldots,k\}$. 
Delete in this execution the \textbf{while} loop iterations that do
not involve the nodes of the game tree $T^{w_l}$ of the subgame $G^{w_l}$ and 
use at the beginning the game tree $T^{w_l}$ instead of
$T$. This way we obtain an execution of Algorithm \ref{alg:BI}
applied to the game $G^{w_l}$ that generates a joint strategy
$s^{w_l}$ in $G^{w_l}$ together with some extensions of
the functions $o_1, \LL, o_n$ to the nodes of the game tree $T^{w_l}$.
By the induction hypothesis $s^{w_l}$ is a subgame perfect equilibrium
in $G^{w_l}$ and $o(w_l) = o(\leaf(s^{w_l}))$.

Consider now the last iteration of the \textbf{while} loop in the
original execution of Algorithm \ref{alg:BI}. At this stage $V =
\{v_0, w_1, \LL, w_k\}$, so before line {\footnotesize{\textbf{2}}} we
have $C(v_0,T) = \{w_1, \LL, w_k\}$.  After line
{\footnotesize{\textbf{3}}} we have $i = \turn(v_0)$ and after line
{\footnotesize{\textbf{4}}} $w$ is such that $w \in \{w_1, \LL,
w_k\}$ and $o_i(w) \geq o_i(w_l)$ for all $l \in \{1, \LL, k\}$.

By the previous conclusion for all $l \in \{1, \LL, k\}$ we have
$o_i(w_l) = o_i(\leaf(s^{w_l}))$, so for all $l \in \{1, \LL, k\}$ we
have $o_i(\leaf(s^{w})) \geq o_i(\leaf(s^{w_l}))$.  After line
{\footnotesize{\textbf{5}}} we have $s_i(v_0) = w$, so by
Lemma~\ref{lm:SPE} $s$ is a subgame perfect equilibrium. Finally,
after line {\footnotesize{\textbf{6}}} we have $o(v_0) = o(\leaf(s^w)) =
o(\leaf(s))$.  \II

\NI $(\La)$ Suppose that $s$ is a subgame perfect equilibrium in $G$.
We show that it can be generated by Algorithm \ref{alg:BI} together
with the extensions of the functions $o_1, \LL, o_n$ to all nodes of
$T$ such that $o(v_0) = o(\leaf(s))$.

Fix an arbitrary $l \in \{1,\ldots,k\}$.  By the assumption on $s$,
the joint strategy $s^{w_l}$ is a subgame perfect equilibrium in the
subgame $G^{w_l}$, so by the induction hypothesis some execution of
Algorithm \ref{alg:BI} applied to the subgame $G^{w_l}$ generates
$s^{w_l}$ together with the extensions of the functions
$o_1, \LL, o_n$ to the nodes of the game tree of $G^{w_l}$ such that
$o(w_l) = o(\leaf(s^{w_l}))$.

Using these $k$ executions of Algorithm \ref{alg:BI} applied to the
subgames $G^{w_1}, \LL, G^{w_k}$ we construct the desired execution of
Algorithm \ref{alg:BI} applied to the game $G$ as follows.  First
we 'glue' these $k$ executions into one but using at the beginning of
the execution the game tree $T$ instead of the game tree of $G^{w_1}$
and using at the beginning of each subsequent execution the current
tree $T$ instead of the game tree of the considered subgame.

After these $k$ executions glued together
$V = \{v_0, w_1, \LL, w_k\}$, so before line
{\footnotesize{\textbf{2}}} we have $C(v,T) = \{w_1, \LL, w_k\}$, in
line {\footnotesize{\textbf{2}}}, $v_0$ is selected and after line
{\footnotesize{\textbf{3}}} we have $i = \turn(v_0)$.

By the induction hypothesis for all $l \in \{1,\ldots,k\}$ we have
$o(w_l) = o(\leaf(s^{w_l}))$, so by Lemma~\ref{lm:SPE} $w = s_i(v_0)$
is a node from $\{w_1, \LL, w_k\}$ such that $o_i(w)$ is maximal. So
in line {\footnotesize{\textbf{4}}} we can select this node $w$, which ensures that the
assignment in line {\footnotesize{\textbf{5}}} is consistent with the original joint strategy
$s$.  Further, the assignment in line {\footnotesize{\textbf{6}}} ensures that
$o(v_0) = o(w) = o(\leaf(s^{w})) = o(\leaf(s))$.
\HB

\begin{corollary}[\cite{Kuhn50}]
  \label{cor:spe}
Every finite extensive game (with perfect information) has a subgame
perfect equilibrium (and hence a Nash equilibrium).
\end{corollary}

We presented backward induction as a nondeterministic algorithm, but
one can go even further. Exploiting the fact that the children of
each node can be dealt with independently, we can present it as an
algorithm that uses nested parallelism. In such an algorithm there is
no need to modify the game tree. Given a non-leaf node $v$ we define
a nondeterministic program $\mathit{Seq}(v)$ by
\II

  $i := \turn(v)$;

  \textbf{choose} $w \in C(v)$ such that $o_i(w)$ is maximal;

  $s_i(v) := w$;

  $o(v) := o(w)$

\II
Then for a preleaf node $v$ we define
$\mathit{Comp}(v)$ as $\mathit{Seq}(v)$
and for each node $v$ that is neither a preleaf or a leaf we define $\mathit{Comp}(v)$
by
\[
[ \|_{w \in C(v)} \mathit{Comp}(w)]; \mathit{Seq}(w)
\]
where $[ \|_{w \in C(v)} \mathit{Comp}(w) ]$ stands for a parallel
composition of the programs $\mathit{Comp}(w)$ for $w \in C(v)$.  So each such
node $v$ is processed with only after its children have been processed
and these children are processed in an arbitrary order.

Then $\mathit{Comp}(v_0)$ is the desired parallel version of the
backward induction algorithm.

\subsection{Special classes of extensive games}

It is natural to study conditions under which an extensive game has a unique
subgame perfect equilibrium.  The following property was introduced in
\cite{Bat97}. We say that an extensive game is \bfe{without relevant
  ties} if for all non-leaf nodes $u$ in $T$ the function $o_i$, where
$\turn(u)=i$, is injective on the leaves of $T^u$. This is more
general than saying that a game is \bfe{generic}, which means that
each $o_i$ is an injective function.

\begin{corollary} \label{cor:generic}
Every finite extensive game without relevant ties has a unique
subgame perfect equilibrium.
\end{corollary}

In particular, every finite generic extensive game has a unique
subgame perfect equilibrium.

\begin{proof}
If a game is without relevant ties, then so is every subgame of
it. This allows us to proceed by induction on the \textit{height} of
the game tree.  Let $G$ be a finite extensive game without relevant
ties over a game tree $T$. If $\mathit{height}(T) = 0$ the claim clearly
holds.  Suppose that $\mathit{height}(T) > 0$.
Let $v$ be the root of $T$ and  $i = \turn(v)$.

By the induction hypothesis for each $w \in C(v)$ there is exactly one
subgame perfect equilibrium $t^w$ in $G^w$. Let
$t=\times_{w \in C(v)} t^w$. Then for different $w, w' \in C(v)$,
$\leaf(t^w)$ and $\leaf(t^{w'})$ are different leaves of the game tree
of $G$. Since $G$ is without relevant ties,
$o_i(\leaf(t^w)) \neq o_i(\leaf(t^{w'}))$.

This means that the function $g: C(v) \to \mathbb{R}$ defined by $g(w)
:= o_i(\leaf(t^w))$ is injective. Consequently the set
$\mathrm{argmax}_{w \in C(v)} o_i(\leaf(t^w))$ has a unique element
and hence by Corollary \ref{lm:SPE2}, $G$ has exactly one subgame
perfect equilibrium.
\end{proof}

Note that the centipede game from Example \ref{exa:centipede} is
generic, so by Corollary \ref{cor:generic} it has exactly one subgame
perfect equilibrium. To determine it we can use the observation there
established, namely that in every Nash equilibrium both players select
$S$ at the nodes $a$ and $b$, respectively. Indeed, by the structure
of the game this observation also holds for every subgame. It follows
that in the unique subgame perfect equilibrium both players select $S$
at all non-leaf nodes.  This counterintuitive form of the subgame
perfect equilibrium in this game is sometimes used to question the
adequacy of this solution concept.

It is also natural to study conditions under which the subgame perfect
equilibria are payoff equivalent. The following theorem is implicit in \cite{MS97}.
The TDI condition was introduced in Section \ref{sec:prelim} when discussing strategic games.

\begin{theorem}
  \label{thm:TDI}
  In every finite extensive game that satisfies the TDI condition all
  subgame perfect equilibria are payoff equivalent.
\end{theorem}

\begin{proof}
  Consider a finite extensive game $G: = (T, \turn, o_1, \LL, o_n)$
  that satisfies the TDI condition.  We proceed by induction on the
  number of nodes in the game tree.  The claim holds when the game
  tree has just one node, since there is then only one subgame perfect
  equilibrium.  Suppose the game tree has more than one node.

  Consider two subgame perfect equilibria $s$ and $t$ in $G$.  Take a
  preleaf $v$ in $T$.  Suppose $s_i(v) = w_1$ and $t_i(v) = w_2$,
  where $i = \turn(v)$. By Corollary \ref{lm:SPE2}
  $w_1, w_2 \in \mathrm{argmax}_{x \in C(v)} o_i(x)$.
\II

\NI
\emph{Case 1}. $\leaf(s) = w_1$ and $\leaf(t) = w_2$.

Take a strategy $s'_i$ that differs from $s_i$ only for the node $v$
to which it assigns $w_2$.  Then $\leaf(s'_i, s_{-i}) = w_2$, so
$o_i(\leaf(s)) = o_i(w_1) = o_i(w_2) = o_i(\leaf(s'_i, s_{-i}))$.
Hence by the TDI property $o(\leaf(s)) = o(\leaf(s'_i, s_{-i}))$, so
$o(\leaf(s)) = o(\leaf(t))$.
\II

\NI
\emph{Case 2}. $\leaf(s) \neq w_1$ or $\leaf(t) \neq w_2$.

Without loss of generality suppose that $\leaf(t) \neq w_2$.
Consider the game $G': = (T', \turn, o_1, \LL, o_n)$ obtained from $G$
by setting $o(v)$ to $o(w_1)$ and by removing all the children of $v$.
So in $G$ the node $v$ is a preleaf, while in $G'$ it is a leaf with
the outcome $o(w_1)$.  $G'$ also satisfies the TDI condition since all
its outcomes are also outcomes of $G$.

Let $s'$ and $t'$ be joint strategies in $G'$ obtained from $s$
and $t$ by dropping $v$ from the domains of $s_i$ and $t_i$.  Then
both $s'$ and $t'$ are subgame perfect equilibria in $G'$.
(We leave the proof of this fact to the reader.)

We have $o(\leaf(s)) = o(\leaf(s'))$ and by assumption the node $v$ does
not lie on the path $\play(t)$, so $\leaf(t) = \leaf(t')$.  Hence
$o(\leaf(s)) = o(\leaf(t))$ by the induction hypothesis.
\end{proof}

\section{Backward induction and common knowledge of rationality}
\label{sec:ckr}

Recall that player's rationality in an extensive game means that his
objective is to maximize his outcome in the game.  Backward induction
is a natural procedure and it is natural to inquire whether it can be
justified by appealing to players' rationality.

In this section we discuss Aumann's result \cite{Aum91} that for a
natural class of extensive games common knowledge of players'
rationality implies that the game reaches the backward induction
outcome.

To formulate this result we introduce first Aumann's approach to
modeling knowledge in the context of extensive games.  Fix a finite
extensive game with no relevant ties $G: = (T, \turn, o_1, \LL, o_n)$
with $T = (V,E,v_0)$. Let $S_1, \LL, S_n$ be the respective sets of
strategies of players $1, \LL, n$.

A \bfe{knowledge system} for $G$ consists of

\begin{itemize}
\item a non-empty set $\Omega$ of \bfe{states},

\item a function $\mathbf{s}: \Omega \to S_1 \times \cdots \times S_n$,

\item for each player $i$ a partition $P_i$ of $\Omega$.  
\end{itemize}

One possible interpretation of a state is that it represents a
`situation', in which complete information about the players'
strategies is available. This information is provided by means of the
function $\mathbf{s}$.

Given a knowledge system, player $i$ does not know the actual state
$\omega$ but he knows the element of the partition $P_i$ to which
$\omega$ belongs.  This interpretation suggests the following
assumption.

Define for player $i$ the function $\mathbf{s}_i: \Omega \to S_i$ by
\[
\mbox{$\mathbf{s}_i(\omega) := s_i$, where $s_i$ is the $i$th component of $\mathbf{s}(\omega)$.}
\]
Then we assume that for each player $i$ the function $\mathbf{s}_i$ is
constant on each element of the partition $P_i$.
% $\mathbf{s}_i(\omega) = \mathbf{s}_i(\omega')$ whenever $\omega$ and
% $\omega'$ belong to the same element of the partition $P_i$.
Intuitively, it means that in the assumed knowledge system each player
knows his strategy.

We first introduce concepts that do not rely on the function
$\mathbf{s}$.  By an \bfe{event} we mean a subset of $\Omega$. For an
event $E$ and player $i$ we define the event $K_i E$ by
\[
  K_i E := \bigcup \: \{L \in P_i \mid L \sse E\}.
\]
Intuitively, $K_i E$ is the event that player $i$ \emph{knows} $E$.

Next, we define
\[
  K E := K^1 E := \bigcap_{i=1}^{n} K_i E,
\]
inductively for $k \geq 1$
\[
  K^{k+1} E := K^k E,
\]
and finally
\[
  C K E := \bigcap_{k=1}^{\infty} K^k E.
\]
Intuitively, $K E$ is the event that all players \emph{know} the event
$E$ and $C K E$ is the event that there is \emph{common knowledge}
of the event $E$ among all players.

Using the function $\mathbf{s}$ we now formalize the notion that player
$i$ is rational.  To start with, given a node $v$ at which player $i$
moves, his strategy $t_i$, and the function $\mathbf{s}_{-i}: \Omega \to S_{-i}$ defined
in the expected way, we denote by
\[
[o_i(\leaf((\mathbf{s}_{-i}, t_i)^v)) > o_i(\leaf(\mathbf{s}^v))]
\]
the event
\[
\{\omega \in \Omega \mid o_i(\leaf((\mathbf{s}_{-i}(\omega), t_i)^v)) > o_i(\leaf(\mathbf{s}(\omega)^v)\}.
\]
It states that in the subgame $G^v$ the outcome for player $i$ is higher if he selects strategy $t_i$ instead
of the strategy he chooses according to $\mathbf{s}$.

Similarly, for a joint strategy $t$ we denote by
\[
  [\mathbf{s} = t]
\]
the event
\[
\{\omega \in \Omega \mid \mathbf{s}(\omega) = t\}.
\]

Recall now that for a player $i$ we denoted by $V_i$ the set of nodes at which
he moves. We define
\[
  R_i := \bigcap_{v \in V_i} \bigcap_{t_i \in S_i} \: \neg K_i
  [o_i(\leaf((\mathbf{s}_{-i}, t_i)^v)) > o_i(\leaf(\mathbf{s}^v))].
\]
where $\neg$ denotes complementation w.r.t.~$\Omega$.
Intuitively, this event states that for all nodes $v$ at which player
$i$ moves and all his strategies $t_i$, player $i$ does not know
whether $t_i$ would yield a higher outcome than the strategy he
chooses according to $\mathbf{s}$.
So $R_i$ is the event formalizing that player $i$ is \emph{rational}. 

Finally, we define
\[
  R := \bigcap_{i=1}^{n} R_i.
\]
Intuitively, $R$ is the event that each player rational.

We still need to formalize the event that the outcome of the game is
prescribed by the backward induction. To this end Aumann assumes that
the game is generic, so that the game has a unique subgame perfect
equilibrium, but thanks to Corollary \ref{cor:generic} it suffices to
assume that the game is without relevant ties.  Then the backward
induction has a unique outcome which is the subgame perfect
equilibrium of the game. Denote the latter by $s^*$.  The
intended event $I$ is then defined by
\[
  I := [\mathbf{s} = s^* ].
\]

We can now state the main result of \cite{Aum91}.

\begin{theorem} \label{thm:ckr}
  Consider a finite extensive game $G$ without relevant ties. Then
  \[
    C K R \sse I.
  \]
\end{theorem}

This inclusion formalizes the announced statement that common
knowledge of players' rationality implies that the backward induction
yields the outcome of the game.

The proof of the theorem relies on a number of simple properties of
the operators $K_i$ and $CK$ that we list without proof in the
following lemma.

\begin{lemma} \label{lem:ck}
\mbox{}  
  \begin{enumerate}[(i)]
  \item $CKE = K_i CKE$.

  \item If $E \sse F$, then $K_i E \sse K_i F$.
    
  \item $K_i E \cap K_i F = K_i(E \cap F)$.

  \item $CKE \sse E$.

  \item $K_i \neg K_i E = \neg K_i E$.

  \item $K_i E \sse E$.
    
  \end{enumerate}
\end{lemma}

Given a joint strategy $s$ and a node $v$ that is not a leaf, we define
\[
s(v) := s_i(v),
\]
where $i = turn(v)$. So if $s$ is the used joint strategy, then $s(v)$
is the move resulting from it at node $v$.  For each such node $v$ we
define the function $\mathbf{s}(v) : \Omega \to V$ in the expected
way.  

We shall also need the following two observations concerning players' knowledge the proofs of which we omit.

\begin{lemma}  \label{lem:sse}
\mbox{}
  \begin{enumerate}[(i)]
  \item For all $t_i \in S_i$, $[\mathbf{s}_i = t_i] \sse K_i [\mathbf{s}_i = t_i]$.
    
  \item For all $v \in V_i$, $I^v \sse K_i I^v$, where $I^v := [\mathbf{s}(v) = s^*(v)]$.
  \end{enumerate}
\end{lemma}

Intuitively, $(i)$ states that if player $i$ chooses the strategy
$t_i$ he knows that he chooses it and $(ii)$ states that if player $i$
chooses the move $s^*(v)$ at the node $v$, then he knows this.
Note that by Lemma \ref{lem:ck}$(vi)$ we can replace in $(i)$ and
$(ii)$ $\sse$ by $=$.

\medskip

\NI
\textbf{Proof of Theorem \ref{thm:ckr}}.

We have $I = \bigcap_{v \in V \setminus Z} I^v$, so it suffices to
prove that for all $v \in V \setminus Z$, $CKR \sse I^v$.  Given two
nodes $v$ and $w$ we write $w < v$ if $w$ is a (possibly indirect)
descendant of $v$.

We proceed by induction w.r.t.~$<$.  Take a node $v$ and suppose that
$CKR \sse I^w$ for all $w < v$.  Let $i = \turn(v)$.  For a joint
strategy $s$ denote by $s^{<v}$ the joint strategy $s^v$ with the pair
$(v, s_i(v))$ removed from $s_i$, and define the function
$\mathbf{s}^{<v}: \Omega \to S_1 \times \cdots \times S_n$ by
\[
\mathbf{s}^{<v}(\omega) := \mathbf{s}(\omega)^{<v}.
\]

By Lemma \ref{lem:ck}$(i)$ and $(ii)$ and the induction hypothesis  $CKR = K_iCKR \sse K_i I^w$ for all $w < v$, so by
Lemma \ref{lem:ck}$(iii)$
\begin{equation}
  \label{equ:a1}
CKR \sse \bigcap_{w < v} K_i I^w = K_i \bigcap_{w < v} I^w = K_i [\mathbf{s}^{<v} = (s^*)^{<v}].  
\end{equation}

Also, by Lemma \ref{lem:ck}$(iv)$ and the definition of $R_i$ with $t_i$ set to $s^*_i$
\begin{equation}
  \label{equ:a2}
CKR \sse R \sse R_i \sse \neg K_i [o_i(\leaf((\mathbf{s}_{-i}, s^*_i)^v)) > o_i(\leaf(\mathbf{s}^v))].
\end{equation}

Further, by Lemma \ref{lem:ck}$(iii)$ and the fact that since $i = \turn(v)$,
\[
  \begin{array}{cl}
    & K_i [\mathbf{s}^{<v} = (s^*)^{<v}] \cap K_i [o_i(\leaf((\mathbf{s}_{-i}, s^*_i)^v)) > o_i(\leaf(\mathbf{s}^v))] \\
   =& K_i [\mathbf{s}^{<v} = (s^*)^{<v} \land o_i(\leaf((\mathbf{s}_{-i}, s^*_i)^v)) > o_i(\leaf(\mathbf{s}^v))] \\
   =& K_i [\mathbf{s}^{<v} = (s^*)^{<v} \land o_i(\leaf((s^*)^v)) > o_i(\leaf((s^*_{-i}, \mathbf{s}_i)^v))] \\
   =& K_i [\mathbf{s}^{<v} = (s^*)^{<v}] \cap K_i [o_i(\leaf((s^*)^v)) > o_i(\leaf((s^*_{-i}, \mathbf{s}_i)^v))],
  \end{array}
\]
so by taking complement w.r.t.~$K_i [\mathbf{s}^{<v} = (s^*)^{<v}]$
\begin{equation}
  \label{equ:a3}
%\[
  \begin{array}{cl}  
    & K_i [\mathbf{s}^{<v} = (s^*)^{<v}] \cap \neg K_i [o_i(\leaf((\mathbf{s}_{-i}, s^*_i)^v)) > o_i(\leaf(\mathbf{s}^v))] \\
   =& K_i [\mathbf{s}^{<v} = (s^*)^{<v}] \cap \neg K_i [o_i(\leaf((s^*)^v)) > o_i(\leaf((s^*_{-i}, \mathbf{s}_i)^v))] \\
  \sse& \neg K_i [o_i(\leaf((s^*)^v)) > o_i(\leaf((s^*_{-i}, \mathbf{s}_i)^v))]. 
  \end{array}       
\end{equation}
%\]

Finally, by (\ref{equ:a1})--(\ref{equ:a3}), the fact that for each node
$v$, $s^v$ is a unique subgame perfect equilibrium in $G^v$, Lemma
\ref{lem:sse}, and Lemma \ref{lem:ck}$(v)$ and $(vi)$
\[
  \begin{array}{cl}
    & CKR \sse K_i [\mathbf{s}^{<v} = (s^*)^{<v}] \cap \neg K_i [o_i(\leaf((\mathbf{s}_{-i}, s^*_i)^v)) > o_i(\leaf(\mathbf{s}^v))] \\
\sse& \neg K_i [o_i(\leaf((s^*)^v)) > o_i(\leaf((s^*_{-i}, \mathbf{s}_i)^v))] =  \neg K_i [\mathbf{s}(v) \neq s^*(v)]  \\
   =& \neg K_i \neg I^v = \neg K_i \neg K_i I^v = \neg \neg K_i I^v = K_i I^v \sse I^v,
  \end{array}
\]
as desired.
\HB
\II

We conclude by the following observation of \cite{Aum91} showing that non-trivial knowledge systems can be easily constructed.

\begin{note}
  For every finite extensive game without relevant ties there exists a knowledge system such that $CKR \neq \ES$.
\end{note}

\begin{proof}
  It suffices to choose $\Omega$ to be a singleton set $\{\omega\}$
  and set $\mathbf{s}(\omega) := s^*$, where $s^*$ is the unique
  subgame perfect equilibrium of the considered game.  Then
  $CKR = \Omega$.
\end{proof}

Aumann's paper dealt with concepts also studied by philosophers and
psychologists. As a result it became highly influential and attracted
wide attention.  In particular Stalnaker pointed out in \cite{Sta96}
that Aumann's notion of rationality involves reasoning about
situations (nodes) that the agent knows will never be reached and
constructed a model in which common knowledge of players' rationality
does not imply that the game reaches the backward induction outcome.

The apparent contradiction between Aumann's and Stalnaker's results
was clarified by Halpern in \cite{Hal01}.  The difference can be
explained by adding to Aumann's knowledge system for an extensive game
one more parameter, a function
\[
  f: \Omega \times V \setminus Z \to \Omega,
\]
that for a given state $\omega$ and a non-leaf node $v$ yields a state
$\omega'$ that is `nearest' (in a well-defined sense) to $v$ and is
such that $v$ is \emph{reached} in $\omega'$, i.e., is such that $v$
lies on $\play(\mathbf{s}(\omega'))$.

Then according to Stalnaker, a player $i$ is \emph{substantively
  rational} in a state $\omega$, if for each node $v \in V_i$ he is
rational in the state $\omega' = f(\omega, v)$, where the latter means that
\[
  \omega' \in \bigcap_{t_i \in S_i} \: \neg K_i
  [o_i(\leaf((\mathbf{s}_{-i}, t_i)^v)) > o_i(\leaf(\mathbf{s}^v))].
\]
Stalnaker's model refers to substantive rationality and not rationality.  
\section{Weak dominance and backward induction}
\label{sec:wd}

Iterated elimination of weakly dominated strategies is defined for
strategic games, so it can be also applied to the strategic forms of
extensive games.  For the class of finite extensive games discussed in
the previous section this procedure is closely related to backward
induction. The aim of this section is to clarify this relation.

The following notion will be needed. Consider a node $w$ in the game
tree of an extensive game $G$ such that $\turn(w) = i$.  We say that a
strategy $s_i$ of player $i$ \bfe{can reach $w$} if for some $s_{-i}$
the node $w$ lies on the path $\play(s_i, s_{-i})$.

We begin by introducing Algorithm \ref{alg:EBI} that is a modification
of the backward induction algorithm \ref{alg:BI} from Section
\ref{sec:spe} in which the input and output are modified and line
{\footnotesize{\textbf{7}}} is added.

{\footnotesize
\begin{algorithm}[htbp]
\caption{\label{alg:EBI}}
\KwIn{A finite extensive game with no relevant ties $G: = (T, \turn, o_1, \LL, o_n)$ with $T = (V,E,v_0)$ 
and $\Gamma(G) = (S_1, \LL, S_n, p_1, \LL, p_n)$.}
\KwOut{The subgame perfect equilibrium $s$ in $G$,
  extensions of the functions $o_1, \LL, o_n$ to all nodes of $T$
  such that $o(v_0) = o(\leaf(s))$, and
a trivial strategic game $(S_1, \LL, S_n, p_1, \LL, p_n)$  that includes $s$.
}
\While{$|V| > 1$}{
  \textbf{choose} $v \in V$ that is a preleaf of $T$;
  
  $i := \turn(v)$;

  \textbf{choose} $w \in C(v,T)$ such that $o_i(w)$ is maximal;

  $s_i(v) := w$;

  $o(v) := o(w)$;

  $S_i := S_i \setminus \{s'_i \in S_i \mid s'_i \textrm{ can reach }
  v \textrm{ and } s'_i(v) \neq w\}$; 
  
  $V := V \setminus C(v,T)$;

  $E := E \cap (V \times V)$;

  $T := (V,E,v_0)$

}  
\end{algorithm}
}

  The following theorem makes precise the mentioned relation between two concepts.
  
\begin{theorem} \label{thm:correspondence}
  Consider a finite extensive game $G$ without relevant ties and Algorithm \ref{alg:EBI} applied to it.
  \begin{enumerate}[(i)]
    
  \item Each strategy removed in line {\footnotesize{\textbf{7}}} is
    weakly dominated in the current version of the strategic game.

  \item The strategic game $(S_1, \LL, S_n, p_1, \LL, p_n)$ which is
    generated upon termination of the algorithm is trivial and
    includes the subgame perfect equilibrium of $G$.
%%   \item Upon termination of the algorithm the resulting strategic game
%%     $(S_1, \LL, S_n, p_1, \LL, p_n)$ is trivial and includes the
%%     subgame perfect equilibrium of $G$.

  \end{enumerate}
\end{theorem}

\begin{proof}
Below $s$ is the subgame perfect equilibrium of the game $G$.
Consider the assertion $I(u)$ defined by
\[
\mbox{$I(u) \equiv \fa t \in S :$ [if $u$ appears in $\play(t)$
then $o(\leaf(s^u)) = o(\leaf(t^u))$]},
\]
where $S = S_1 \times \cdots \times S_n$.

First notice that if during an execution of the algorithm for some node
$u$ the assertion $I(u)$ becomes true, then it remains true. The
reason is that the considered set $S$ of joint strategies never increases.

We now show that after each loop iteration $I(v)$ holds, where $v$ is
the node being dealt with in current loop iteration and 
$S$ refers to the current set of joint strategies. Fix an execution of
the algorithm.  We proceed by induction on the order in which the
nodes of $T$ are selected in line {\footnotesize{\textbf{2}}}.
Consider first a preleaf $v$ in the original game tree $T$. Take
$t \in S$ such that $v$ appears in $\play(t)$ and let $i =
\turn(v)$. Then \II

\begin{tabular}{ll}
      & $o(\leaf(t^v))$ \\
    = &  \ \ \{ by line {\footnotesize{\textbf{7}}} $t_i(v) = w$\} \\
      & $o(w)$ \\
    = &  \ \ \{ by line {\footnotesize{\textbf{5}}} $s_i(v) = w$\} \\
      & $o(\leaf(s^v))$.
\end{tabular}
\II

Next, consider a node $v$ in the original game tree $T$ selected in
line {\footnotesize{\textbf{2}}} that is neither a preleaf nor a
leaf and consider the program state after the current loop iteration.
Then both $i = \turn(v)$ and $s_i(v) = w$.

By the order in which the nodes are selected in line
{\footnotesize{\textbf{2}}}, all nodes $u \in C(v)$ have been dealt
with in earlier loop iterations. So by the induction hypothesis $I(u)$
holds for $u \in C(v)$. In particular $I(w)$ holds. 
Take $t \in S$ such that $v$ appears in $\play(t)$. Then
\II

\begin{tabular}{ll}
      & $o(\leaf(t^v))$ \\
  = &  \ \ \{ by line {\footnotesize{\textbf{7}}} $t_i(v) = w$\} \\
      & $o(\leaf(t^w))$ \\
    = &  \ \ \{ $I(w))$ \} \\
      & $o(\leaf(s^w))$ \\ 
    = &  \ \ \{ by line {\footnotesize{\textbf{5}}} $s_i(v) = w$\} \\
      & $o(\leaf(s^v))$.
\end{tabular}
\II

\NI
$(i)$
Fix the program state after an arbitrary loop iteration of the
algorithm with the current values of $v, w$, $i$,
$S_1, \LL, S_n$. In particular $\turn(v) = i$.

We first prove that for $u \in C(v)$, $u \neq w$
\begin{equation}
  \label{equ:uw}
o_i(\leaf(s^w)) > o_i(\leaf(s^u)).
\end{equation}

Let $t_i$ be the strategy of player $i$ that differs from $s_i$ only
for the node $v$ to which it assigns $u$.  We have $s_i(v) = w$ and
by definition $s^v$ is a Nash equilibrium in the subgame $G^v$, so
\[
o_i(\leaf(s^w)) = o_i(\leaf(s^v)) \geq o_i(\leaf(t_i^v, s_{-i}^v)) = o_i(\leaf(t_i^u, s_{-i}^u)) = o_i(\leaf(s^u)).
\]
But $\turn(v) = i$, $u, w \in C(v)$, $u \neq w$, and $G^v$ is without relevant
ties, so (\ref{equ:uw}) follows.

Take now a strategy $s'_i$ removed in line {\footnotesize{\textbf{7}}}
and suppose $s'_i(v) = u$.  Consider the strategy $t_i$ that differs
from $s'_i$ only for the node $v$ to which it assigns $w$ (i.e.,
$s_i(v)$).  We claim that after line {\footnotesize{\textbf{7}}} $t_i$
weakly dominates $s'_i$ in the current version of
$(S_1, \LL, S_n, p_1, \LL, p_n)$.

Take some $t_{-i} \in S_{-i}$.  If $v \not\in \play(t_i, t_{-i})$,
then $v \not\in \play(s'_i, t_{-i})$, so
$\play(t_i, t_{-i}) = \play(s'_i, t_{-i})$ and hence
\[
p_i(t_i, t_{-i}) = 
o_i(\leaf(t_i, t_{-i}) = o_i(\leaf(s'_i, t_{-i}) = p_i(s'_i, t_{-i}).
\]

If $v \in \play(t_i, t_{-i})$, then also $w \in \play(t_i, t_{-i})$, 
$v \in \play(s'_i, t_{-i})$ and $u \in \play(s'_i, t_{-i})$, where, recall, $s'_i(v) = u$.
Since $u, w \in C(v)$, these two nodes have been dealt with in earlier loop iterations.
So both $I(u)$ and $I(w)$ hold.

Hence 
\[
  p_i(t_i, t_{-i}) = o_i(\leaf(t_i, t_{-i})) = o_i(\leaf((t_i, t_{-i})^w)) = o_i(\leaf(s^w))
\]
and 
\[
  p_i(s'_i, t_{-i}) = o_i(\leaf(s'_i, t_{-i})) = o_i(\leaf((s'_i, t_{-i})^u)) = o_i(\leaf(s^u)).
\]
So $p_i(t_i, t_{-i}) > p_i(s'_i, t_{-i})$ by (\ref{equ:uw}).

Now, $t_i$ does not need to be a strategy from $S_i$ but thanks to
Lemma \ref{lem:wd} we can conclude that a strategy $t'_i$ in $S_i$
exists that weakly dominates $s'_i$ in the current version of
$(S_1, \LL, S_n, p_1, \LL, p_n)$.
\II

\NI
$(ii)$
Upon termination of the algorithm $I(v_0)$, i.e.,
$\fa t \in S : o(\leaf(s)) = o(\leaf(t))$ holds. This means that upon
termination the final game $(S_1, \LL, S_n, p_1, \LL, p_n)$ is
trivial.  Further, for each player each strategy removed in line
{\footnotesize{\textbf{7}}} differs from his strategy in the subgame
equilibrium, which means that the final game includes this
equilibrium.
\end{proof}
\II

This theorem shows that every finite extensive game $G$ without
relevant ties can be solved by an IEWDS. Recall from Corollary
\ref{cor:generic} that such a game has a unique subgame perfect
equilibrium.
%% A simple example (see \cite[page 109]{OR94}) shows that
%% some iterations of IEWDS may remove the subgame perfect
%% equilibrium. Also one cannot stipulate that the final game has the
%% subgame perfect equilibrium as the unique joint strategy.
However, not all instances of the IEWDS behave the desired way.  The
following example, taken from \cite[page 109]{OR94}, shows that 
  % there exists extensive games in which an IEWDS can remove all subgame
  % perfect equilibria. Also, one cannot stipulate that the final game
  % has the subgame perfect equilibrium as the unique joint strategy.
some generic extensive games can be solved by an IEWDS that removes the unique
subgame perfect equilibrium. This explains why in line
{\footnotesize{\textbf{7}}} in Algorithm \ref{alg:EBI} only specific
  weakly dominated strategies are removed.

\begin{figure}[htbp]
  \centering
  \begin{minipage}{.7\textwidth}
    \centering
\tikzstyle{agent}=[circle,draw=black!80,thick, minimum size=1em]

\tikzstyle{changeblue}=[draw=blue!80,fill=blue!30,thick]
\tikzstyle{changered}=[draw=red!80,fill=red!30,thick]
\tikzstyle{changegreen}=[draw=green!80,fill=green!30,thick]
 \begin{tikzpicture}[auto,>=latex', align=center]
   \node[label=above:{\small }](a1)    {$1$};

   \node[below left of=a1, node distance=2cm,label=right:{\small }](b1){$2$};

   \node[below right of=a1, node distance=2cm,label=right:{\small }](b2){$(3,3)$};

   \node[below left of=b1, node distance=2cm,label=right:{\small }](c1){$1$};

   \node[below right of=b1, node distance=2cm,label=right:{\small }](c2){$(1,1)$};

   \node[below left of=c1, node distance=2cm,label=right:{\small }](d1){$(2,0)$};
   \node[below right of=c1, node distance=2cm,label=right:{\small }](d2){$(0,2)$};

   \draw (a1) -- (b1) node[midway, above] {{\small $A$}};
   \draw (a1) -- (b2) node[midway, above] {{\small $B$}};

   \draw (b1) -- (c1) node[midway, above] {{\small $C$}};
   \draw (b1) -- (c2) node[midway, above] {{\small $D$}};

   \draw (c1) -- (d1) node[midway, above] {{\small $E$}};
   \draw (c1) -- (d2) node[midway, above] {{\small $F$}};

   %% \node[right of=b3,node distance=1.5cm] (b4){$\cdots$};

   %%    \node[agent,left of=b2, node distance=2cm,
   %%      label=right:{\small }, onslide={<1->{changeblue}}](b1){$u_1$};
 \end{tikzpicture}
 \end{minipage}%
\begin{minipage}{.3\textwidth}
 \begin{game}{4}{2}
      & $C$               & $D$\\
   $\mathit{AE}$   &$\phantom{-}2, 0$    &$\phantom{-}1, 1$\\
   $\mathit{AF}$   &$\phantom{-}0, 2$    &$\phantom{-}1, 1$\\
   $\mathit{BE}$   &$\phantom{-}3, 3$    &$\phantom{-}3, 3$\\
   $\mathit{BF}$   &$\phantom{-}3, 3$    &$\phantom{-}3, 3$\\
 \end{game}
\end{minipage}
\caption{\label{fig:SPEIEWDS} An extensive game (left) and its associated strategic game (right).}
\end{figure}

\begin{example}
  Consider the two-player generic extensive game and its associated strategic
  game given in Figure~\ref{fig:SPEIEWDS}. In the figure, the nodes
  are labelled with the player whose turn it is to move.

  Consider now an IEWDS that consists of the following sequence of
  elimination of weakly dominated strategies:
  $\mathit{AE}, D, \mathit{AF}$. The resulting trivial subgame has two
  joint strategies $(\mathit{BE}, C)$, $(\mathit{BF}, C)$.  So this
  instance of the IEWDS eliminates $(\mathit{BE}, D)$, which is the
  unique subgame perfect equilibrium.  \HB
\end{example}

\section{Weak dominance and strictly competitive games}

We now continue an account of iterated elimination of weakly dominated
strategies.  In Theorem \ref{thm:correspondence} we showed that each
finite extensive game without relevant ties can be solved by an IEWDS.
A natural question is whether we can extend this result to arbitrary
finite extensive games.  The following example taken from \cite[pages
109-110]{OR94} shows that this fails to be the case already for
two-player games.

\label{sec:sc}
\begin{figure}[htbp]
  \centering
  \begin{minipage}{.7\textwidth}
    \centering
\tikzstyle{agent}=[circle,draw=black!80,thick, minimum size=1em]

\tikzstyle{changeblue}=[draw=blue!80,fill=blue!30,thick]
\tikzstyle{changered}=[draw=red!80,fill=red!30,thick]
\tikzstyle{changegreen}=[draw=green!80,fill=green!30,thick]
 \begin{tikzpicture}[auto,>=latex', align=center]
   \node[label=above:{\small }](a1)    {$1$};

   \node[below left=.6cm and 1.5cm of a1,label=right:{\small }](b1){$2$};

%   \node[below right of=a1, node distance=2cm,label=right:{\small }](b2){$(3,3)$};

      \node[below right=.6cm and 1.5cm of a1,label=right:{\small }](b2){$1$};

   \node[below left of=b1, node distance=2cm,label=right:{\small }](c1){$(0,0)$};

   \node[below right of=b1, node distance=2cm,label=right:{\small }](c2){$(2,0)$};

   \node[below left of=b2, node distance=2cm,label=right:{\small }](c3){$(1,1)$};
   \node[below right of=b2, node distance=2cm,label=right:{\small }](c4){$(0,0)$};

   \draw (a1) -- (b1) node[midway, above] {{\small $A$}};
   \draw (a1) -- (b2) node[midway, above] {{\small $B$}};

   \draw (b1) -- (c1) node[midway, above] {{\small $L$}};
   \draw (b1) -- (c2) node[midway, above] {{\small $R$}};

   \draw (b2) -- (c3) node[midway, above] {{\small $C$}};
   \draw (b2) -- (c4) node[midway, above] {{\small $D$}};

   %% \node[right of=b3,node distance=1.5cm] (b4){$\cdots$};

   %%    \node[agent,left of=b2, node distance=2cm,
   %%      label=right:{\small }, onslide={<1->{changeblue}}](b1){$u_1$};
 \end{tikzpicture}
 \end{minipage}%
\begin{minipage}{.3\textwidth}
 \begin{game}{4}{2}
      & $L$               & $R$\\
   $\mathit{AC}$   &$\phantom{-}0, 0$    &$\phantom{-}2, 0$\\
   $\mathit{AD}$   &$\phantom{-}0, 0$    &$\phantom{-}2, 0$\\
   $\mathit{BC}$   &$\phantom{-}1, 1$    &$\phantom{-}1, 1$\\
   $\mathit{BD}$   &$\phantom{-}0, 0$    &$\phantom{-}0, 0$\\
 \end{game}
\end{minipage}
\caption{\label{fig:IEWDSnontrivial} An extensive game (left) and its associated strategic game (right).}
\end{figure}

\begin{example}
  Consider the two-player extensive game and its associated
  strategic game given in Figure \ref{fig:IEWDSnontrivial}. In the
  game tree, the nodes are labelled with the player whose turn it is
  to move. For this game there is just one instance of IEWDS which
  consists of eliminating the strategy $\mathit{BD}$. The
  resulting subgame is not trivial, so no instance of IEWDS can
  solve this extensive game.
  \HB
\end{example}

On the other hand, as shown in \cite{Ewe02}, finite extensive zero-sum
games can be solved by an IEWDS in which at each step all weakly
dominated strategies are removed.  The aim of this section is to
present this result for the larger class of finite extensive strictly
competitive games for which the same proof remains valid.

From now on, given a strategic game $H$ we denote by $H^1$ a subgame
of $H$ obtained by the elimination of \emph{all} strategies that are
weakly dominated in $H$, and put $H^{0} := H$ and
$H^{k+1} := (H^k)^1$, where $k \geq 1$.  So, in contrast to Sections
\ref{sec:prelim} and \ref{sec:aux} each $H^k$ is now uniquely defined.

Below for a strategic game $H$ we denote by $H_i$ the set of
strategies of player $i$.  Also, for a finite extensive game $G$ we
write $\Gamma^{k}(G)$ instead of $(\Gamma(G))^k$, $\Gamma_i(G)$
instead of $(\Gamma(G))_i$, and $\Gamma^k_i(G)$ instead of
$(\Gamma^k(G))_i$.  In particular $\Gamma^{0}(G) = \Gamma(G)$.

Further, for a finite strictly competitive strategic game $H = (S_1, S_2, p_1, p_2)$
we define for each player $i$ 
\[
\begin{array}{l}
p_i^{\pmax}(H) := \max_{s \in S} p_i(s), \\
\wmax_i(H) : =\{s_i \in S_i \mid \fa s_{-i} \in S_{-i} \: p_i(s_i,s_{-i})=p_i^{\pmax}(H)\}, \\
\lose_{-i}(H) = \{s_{-i} \in S_{-i} \mid \exists s_{i} \in S_{i} \: p_{i}(s_{i},s_{-i}) = p^{\pmax}_{i}(H)\}. 
\end{array}
\]
So $p_i^{\pmax}(H)$ is the maximal payoff player $i$ can receive in the game $H$,
$\wmax_i(H)$ is the set of strategies of player $i$ for which he
always gets $p_i^{\pmax}(H)$, while $\lose_{-i}(H)$ is the set of
strategies of player $-i$ for which his opponent $i$ can get
his maximally possible payoff $p_i^{\pmax}(H)$.

The following lemma, with a rather involved proof, is crucial.

\begin{lemma}
  \label{lem:scLose}
  Let $G$ be a finite strictly competitive extensive game.
    For all $i \in \{1,2\}$ and for
  all $k \geq 0$, if $\wmax_i(\Gamma^k(G)) = \emptyset$ then 
    $\lose_{-i}(\Gamma^k(G)) \cap \Gamma^{k+2}_{-i}(G) = \ES$.
  \end{lemma}

  This lemma implies that if for all $i \in \{1,2\}$,
  $\wmax_i(\Gamma^k(G)) = \emptyset$ then two further rounds of
  eliminations of weakly dominated strategies remove from
  $\Gamma^k(G)$ at least two outcomes.

\begin{proof}
Fix $i$ and $k$ and suppose $\wmax_i(\Gamma^k(G)) = \emptyset$. So for all 
$s_i \in \Gamma^k_i(G)$ we have
$\min_{s_{-i} \in \Gamma^k_{-i}(G)} p_i(s_i, s_{-i}) < p_i^{\pmax}(H)$, and
hence $\maxmin_i(\Gamma^k(G)) < p_i^{\pmax}(\Gamma^k(G))$.

By Corollary \ref{cor:spe} the strategic game $\Gamma(G)$ has a Nash
equilibrium.  By the repeated application of Corollary \ref{cor:h1} we
have $\maxmin_i(\Gamma(G))=\maxmin_i(\Gamma^k(G))$. Therefore
  \begin{equation}
    \label{equ:win}
\maxmin_i(\Gamma(G)) < p_i^{\pmax}(\Gamma^k(G)).
  \end{equation}

Take now $s_{-i} \in \lose_{-i}(\Gamma^k(G))$. We need to prove that
$s_{-i} \not\in \Gamma^{k+2}_{-i}(G)$.

For some $s_i \in \Gamma^k_i(G)$
we have $p_i(s_i,s_{-i})=p_i^{\pmax}(\Gamma^k(G))$.
By Lemma \ref{lem:wd} we can assume that $s_i \in \Gamma^{k+1}_i(G)$.
Consider now the path $\mathit{play}(s_i,s_{-i})$. Then

\begin{itemize}

\item by (\ref{equ:win}) for the first node $u$ lying on
  $\mathit{play}(s_i,s_{-i})$ (so the root),

  $\maxmin_i(\Gamma(G^u)) < p_i^{\pmax}(\Gamma^k(G))$,

\item for the last node $u$ lying on $\mathit{play}(s_i,s_{-i})$ (so the leaf), 

$p_i^{\pmax}(\Gamma^k(G)) = p_i(s_i,s_{-i}) = o_i(u) = \maxmin_i(\Gamma(G^u))$.
\end{itemize}

% we have by (\ref{equ:win})
% $\maxmin_i(\Gamma(G^u)) < p_i^{\pmax}(\Gamma^k(G))$ and for the last
% node $v$ (so the leaf) we have
% $\maxmin_i(\Gamma(G^v)) = p_i^{\pmax}(\Gamma^k(G))$.  

So for some adjacent nodes $u, w$ lying on the path
$\mathit{play}(s_i,s_{-i})$ 
% we have
% $\maxmin_i(\Gamma(G^u)) < p_i^{\pmax}(\Gamma^k(G))$ and
% $\maxmin_i(\Gamma(G^v)) \geq p_i^{\pmax}(\Gamma^k(G))$, i.e.,
\begin{equation}
  \label{equ:pmax}
  \maxmin_i(\Gamma(G^u)) <
p_i^{\pmax}(\Gamma^k(G)) \leq \maxmin_i(\Gamma(G^w)).
\end{equation}

Further, if for some adjacent nodes $u', w'$ lying on the path
$\mathit{play}(s_i,s_{-i})$ we have $\turn(u') =i$ and
$p_i^{\pmax}(\Gamma^k(G)) \leq \maxmin_i(\Gamma(G^{w'}))$, then
$\maxmin_i(\Gamma(G^{u'})) = \maxmin_i(\Gamma(G^{w'}))$.
So $\turn(u)=-i$ and $s_{-i}(u)=w$.

If $s_{-i} \notin \Gamma^{k+1}_{-i}(G)$ then 
$s_{-i} \notin \Gamma^{k+2}_{-i}(G)$.
So suppose $s_{-i} \in \Gamma^{k+1}_{-i}(G)$. 
We prove that then $s_{-i}$ is weakly dominated in
$\Gamma^{k+1}(G)$. The dominating strategy is obtained in two steps.

By Corollary \ref{cor:spe} the game $\Gamma(G^u)$ has a Nash equilibrium
$s^*$. First, we introduce the strategy $t_{-i} \in \Gamma_{-i}(G)$
defined as follows:
\[
  t_{-i}(x) : =\begin{cases}
    s_{-i}(x) &\text{ if } x \text{ not in } T^u,\\
    s^*_{-i}(x) &\text{ if } x \text{ in } T^u,
  \end{cases}
\]
where $\turn(x)=-i$ and $T$ is the game tree of $G$.

We now establish two claims relating $t_{-i}$ to $s_{-i}$.
%% $t_{-i}$ weakly dominates $s_{-i}$ in
%% $\Gamma^{k+1}(G)$ by showing the following:
\II

\NI
\textbf{Claim 1}.
$\fa s_i' \in \Gamma^{k+1}_i(G) : p_{-i}(s_i',t_{-i}) \geq p_{-i}(s_i',s_{-i})$.
\II
    
\NI
\emph{Proof.}
Suppose by contradiction that there exists $s_i' \in \Gamma^{k+1}_i(G)$ such that
$p_{-i}(s_i',t_{-i}) < p_{-i}(s_i',s_{-i})$.  The strategy
$t_{-i}$ differs from $s_{-i}$ only on the nodes in $T^u$, so
the difference in the payoffs implies that $u$ appears both in
$\mathit{play}(s'_i,t_{-i})$ and $\mathit{play}(s'_i,s_{-i})$.
This implies 
\begin{equation}
  \label{equ:ibar}
\maxmin_{-i}(\Gamma(G^u)) \leq 
p_{-i}((s_i')^u,s^*_{-i}) = p_{-i}(s_i',t_{-i}) < p_{-i}(s_i',s_{-i}).
\end{equation}

% $G^u$ is strictly competitive, so
% \begin{equation}
%   \label{equ:weakbetter}
%   p_{i}(s_i',s_{-i}) < \maxmin_{i}(\Gamma(G^u)).
% \end{equation}

% By definition of $t_{-i}$ we have
% $p_{-i}(s_i',t_{-i}) \geq \maxmin_{-i}(\Gamma(G^u))$ and by
% assumption $p_{-i}(s_i',t_{-i}) <
% p_{-i}(s_i',s_{-i})$. Therefore $p_{-i}(s_i',s_{-i}) >
% \maxmin_{-i}(\Gamma(G^u))$.

By Theorem \ref{thm:scminimax} $s^*_{-i}$
is a security strategy of player $-i$ in the game $\Gamma(G^u)$.
Further, the node $u$ appears in $\mathit{play}(s_i',s_{-i})$, so
$s'' := (s_i',s_{-i})^u$ is a joint strategy in $G^u$.
This and (\ref{equ:ibar}) imply
\[
p_{-i}(s^*) = \maxmin_{-i}(\Gamma(G^u)) < p_{-i}(s_i',s_{-i}) = p_{-i}(s''),
\]
so by (\ref{equ:iff}), the fact that $G^u$ is strictly competitive, and
Theorem \ref{thm:scminimax}
\begin{equation}
  \label{equ:weakbetter}
  p_{i}(s_i',s_{-i}) = p_{i}(s'') < p_{i}(s^*) = \maxmin_{i}(\Gamma(G^u)).
\end{equation}

Next we introduce the strategy $t_i \in \Gamma_i(G)$ defined as follows
(recall that $w = s_{-i}(u)$):
\[
  t_{i}(x) :=\begin{cases}
    s_{i}'(x) &\text{ if } x \text{ not in } T^w,\\
    t^*_{i}(x) &\text{ if } x \text{ in } T^w,
\end{cases}
\]
where $\turn(x)=i$ and $t^*_{i}$ is a
security strategy of player $i$ in the game $\Gamma(G^w)$.

We now establish two claims relating $t_{i}$ to $s'_{i}$:
\begin{equation}
  \label{equ:st1}
\fa s_{-i}' \in \Gamma_{-i}^k(G) : p_i(t_i,s_{-i}') \geq  p_i(s'_i,s'_{-i}),
\end{equation}
\begin{equation}
  \label{equ:st2}
p_i(t_i,s_{-i}) >  p_i(s'_i,s_{-i}).
\end{equation}

To establish (\ref{equ:st1}) consider any strategy
$s_{-i}' \in \Gamma_{-i}^k(G)$. By the definition of $t_i$, if $w$
does not appear in $\mathit{play}(t_i,s'_{-i})$ then
$p_i(t_i,s'_{-i}) = p_i(s_i',s'_{-i})$. 
So suppose $w$ appears in $\mathit{play}(t_i,s'_{-i})$.
By the definition of $t_i$, (\ref{equ:pmax}), and the fact that both
$s'_i$ and $s_{-i}'$ are strategies in $\Gamma^k(G)$
\[
p_i(t_i,s_{-i}') = p_i(t^{*}_i,(s_{-i}')^w) \geq \maxmin_i(\Gamma(G^w)) \geq
p_i^{\pmax}(\Gamma^k(G)) \geq p_i(s_i',s_{-i}').
\]

To establish (\ref{equ:st2}) recall that we noted already that $u$
appears in $\mathit{play}(s'_i,s_{-i})$.  Since $\turn(u)=-i$
and $s_{-i}(u)=w$, also $w$ appears in
$\mathit{play}(s_i',s_{-i})$.  The strategy $t_{i}$ differs from
$s'_{i}$ only on the nodes in $T^w$, so $w$ apears in
$\mathit{play}(t_i,s_{-i})$, as well.  Therefore by the definition
of $t_i$, (\ref{equ:pmax}) and (\ref{equ:weakbetter})
\[
  p_i(t_i,s_{-i}) = p_i(t^{*}_i,(s_{-i})^w) \geq
  \maxmin_i(\Gamma(G^w)) > \maxmin_i(\Gamma(G^u)) >
  p_i(s_i',s_{-i}).
\]

By Lemma \ref{lem:wd} there exists $t'_i \in \Gamma^{k}_i(G)$ such that
$p_i(t'_i,s'_{-i}) \geq  p_i(t_i,s'_{-i})$ for all $s'_{-i} \in \Gamma^{k}_{-i}(G)$.
This, together with (\ref{equ:st1}) and (\ref{equ:st2}) implies that
$s_i'$ is weakly dominated by $t'_i$ in $\Gamma^k(G)$. Hence
$s_i' \notin \Gamma^{k+1}_i(G)$, which yields a contradiction.
\HB
\II

\NI
\textbf{Claim 2.}
$p_{-i}(s_i,t_{-i}) > p_{-i}(s_i,s_{-i})$.
\II

\NI
\emph{Proof.}
The node $u$ appears in $\mathit{play}(s_i,s_{-i})$,
so by Theorem \ref{thm:scminimax} and (\ref{equ:pmax})
\[
p_i(s^*) = \maxmin_i(\Gamma(G^u)) < p_i^{\pmax}(\Gamma^k(G)) = p_i(s_i,s_{-i}) = p_i(s^u_i,s^u_{-i}),
\]
and consequently by Theorem \ref{thm:scminimax}$(i)$ and the fact that
$G^u$ is strictly competitive
\begin{equation}
  \label{equ:gamma1}
\maxmin_{-i}(\Gamma(G^u)) = p_{-i}(s^*) >
  p_{-i}(s^u_i,s^u_{-i}) = p_{-i}(s_i,s_{-i}).  
\end{equation}

Further, the strategy $t_{-i}$ differs from $s_{-i}$ only on the nodes in
$T^u$, so $u$ appears in $\mathit{play}(s_i,t_{-i})$, as well.
Hence
\begin{equation}
  \label{equ:gamma2}
  p_{-i}(s_i,t_{-i})  = p_{-i}(s^u_i,s^*_{-i}) 
  \geq \maxmin_{-i}(\Gamma(G^u)).
\end{equation}
Combining (\ref{equ:gamma1}) and (\ref{equ:gamma2}) we get the claim.
\HB
\II

By Lemma \ref{lem:wd} there exists
$t'_{-i} \in \Gamma^{k+1}_{-i}(G)$ such that
$p_{-i}(s'_i,t'_{-i}) \geq p_{-i}(s'_i,t_{-i})$ for all
$s'_i \in \Gamma^{k+1}_i(G)$. 
We conclude by Claims 1 and 2 that
$s_{-i}$ is weakly dominated by $t'_{-i}$ in
$\Gamma^{k+1}(G)$. 
Therefore $s_{-i} \notin \Gamma^{k+2}_{-i}(G)$, as desired.
\end{proof}

We can now establish the announced result.
\begin{theorem}
\label{thm:sciewds}
Let $G$ be a finite strictly competitive
extensive game with at most $m$ outcomes. Then $\Gamma^{m-1}(G)$ is a trivial
game.
\end{theorem}

\begin{proof}
We prove a stronger claim, namely that for all $m \geq 1$ and $k \geq 0$
if $\Gamma^{k}(G)$ has at most $m$ outcomes, then $\Gamma^{k+m-1}(G)$ is a
trivial game.  

We proceed by induction on $m$. We can assume that $m > 1$. 
For $m = 2$ the claim follows by Lemma \ref{lem:scTwo}.
Take $m > 2$.
\II

\NI
\emph{Case 1.} For some $i \in \{1,2\}$, $\wmax_i(\Gamma^k(G)) \neq \emptyset$.

For player $i$ every strategy $s_i \in \wmax_i(\Gamma^k(G))$ weakly
dominates all strategies $s_i' \notin \wmax_i(\Gamma^k(G))$. So in
$\Gamma^{k+1}(G)$ the set of strategies of player $i$ equals
$\wmax_i(\Gamma^k(G))$ and hence $p_i^{\pmax}(\Gamma^k(G))$ is his unique payoff
in this game. By (\ref{equ:iff}) $\Gamma^{k+1}(G)$, and hence also
$\Gamma^{k+m-1}(G)$, is a trivial game. 
\II

\NI
\emph{Case 2.} For all $i \in \{1,2\}$, $\wmax_i(\Gamma^k(G)) = \emptyset$.

Take joint strategies $s$ and $t$ such that
$p_1(s)=p_1^{\max}(\Gamma^k(G))$ and $p_2(t)=p_2^{\max}(\Gamma^k(G))$.
Since $m > 1$
(\ref{equ:iff}) implies that
the outcomes $(p_1(s), p_2(s))$ and $(p_1(t), p_2(t))$ are different.

We have $s_2 \in \lose_{2}(\Gamma^k(G))$ and $t_1 \in \lose_{1}(\Gamma^k(G))$.
Hence by Lemma \ref{lem:scLose} for no joint strategy $s'$ in $\Gamma^{k+2}(G)$
we have $p_1(s')=p_1^{\max}(\Gamma^k(G))$ or  $p_2(s')=p_2^{\max}(\Gamma^k(G))$.

So $\Gamma^{k+2}(G)$ has at most $m-2$ outcomes.  By the induction
hypothesis $\Gamma^{k+m-1}(G)$ is a trivial game.
\end{proof}

\section{Weak acyclicity}
\label{sec:wacyclic}

By Theorem \ref{cor:spe} every finite extensive game has a Nash
equilibrium. A natural question is whether we can strengthen this
result to show that finite extensive games have the FIP.  The
following example adapted from \cite{Mil13} shows that even for
simplest extensive games the answer is negative.

\begin{example}
  \label{exa:extNoFIP}
  Consider the extensive form game given in
  Figure~\ref{fig:prisoner}. Following the convention introduced in
  Example~\ref{exa:extMP}, the strategies for player 1 are $C$ and
  $D$, while the strategies of player 2 are $\mathit{CC}, \mathit{CD},
  \mathit{DC}$ and $\mathit{DD}$.

  Then the following improvement sequence generates an infinite
  improvement path in this game:
\II
  
% \[
%   (D,\underline{\mathit{DC}}) \to (\underline{D},\mathit{CD}) \to
%   (C,\underline{\mathit{CD}}) \to (\underline{C},\mathit{DC}) \to
%   (D,\underline{\mathit{DC}}).
% \]
\begin{tabular}{ccccccccc}
  $(D,\underline{\mathit{DC}})$ &  $\to$ & $(\underline{D},\mathit{CD})$ &  $\to$ &
  $(C,\underline{\mathit{CD}})$ &  $\to$ & $(\underline{C},\mathit{DC})$ &  $\to$ &
  $(D,\underline{\mathit{DC}})$ \\
  $(3,0)$ &  & $(1,1)$ &  &  $(2,2)$ & & $(0,3)$ &  & $(3,0)$
  % the corresponding vectors of outcomes: (3,0), (1,1), (2,2), (0,3) and (3,0).
\end{tabular}
\II

\NI
For convenience of the reader in each joint strategy we underlined the
strategy which is not a best response and listed the corresponding outcomes.
\HB
\end{example}

However, a weaker result, due to \cite{Kuk02}, does hold. It implies
that every finite extensive game has a Nash equilibrium, a result
established earlier, in Corollary \ref{cor:spe}.

\begin{theorem} \label{thm:Kuk}
  Every finite extensive game is weakly acyclic.
\end{theorem}
\textbf{Proof.}
We prove the claim by defining a weak potential.
Take a finite extensive game $G:= (T, \turn, o_1, \LL, o_n)$, with
$T := (V,E,v_0)$ and let $S$ be the set of joint strategies.

Consider first a function $R: S \times V \to \{0,1\}$
defined as follows:
\[
\begin{array}{l}
  R(s,v) :=
  \begin{cases}
    1 \ \mbox{if $s^v_i$ is a best response to $s^v_{-i}$ in the subgame $G^v$} \\
    0 \ \mbox{otherwise,} \\
  \end{cases}
\end{array}
\]
where $i = \turn(v)$.

Let now $L := (v_1, \LL, v_k)$ be a list of the nodes from
$V$ such that each node appears after all of its children in 
$T$. For example,
\[
((2,2), (0,3), (3,0), b, (1,1), c, a)
\]
is such a list of the nodes of the tree from Figure \ref{fig:prisoner},
where we identify each leaf with the corresponding outcome.

Finally define the function $P: S \to \{0,1\}^k$ by putting
\[
  P(s) := (R(s,v_1), \LL, R(s,v_k)),
\]
where, recall, $L= (v_1, \LL, v_k)$, and consider the strict
lexicographic ordering $<_{lex}$ over $\{0,1\}^k$.  We show that $P$
is a weak potential w.r.t.~this ordering and appeal to the
Weakly Acyclic Theorem \ref{thm:weakly}.

So consider a joint strategy $s$ in $G$ that is not a Nash
equilibrium.  Take a player $i$ such that $s_i$ is not a best response
to $s_{-i}$. Let $t_i$ be a best response of player $i$ to $s_{-i}$
and let $t = (t_i, s_{-i})$.  Define $s'_i$ as the modification of
$t_i$ such that its values on the nodes not lying on $play(t)$ are the
values provided by $s_i$. More formally, for all nodes $v$ such that
$\turn(v) = i$ we put
\[
\begin{array}{l}
s'_i(v) :=
  \begin{cases}
    t_i(v)         & \mbox{if $v$ lies on $play(t)$} \\
    s_i(v) & \mbox{otherwise} \\
  \end{cases}
\end{array}
\]
Let $s' = (s'_i, s_{-i})$. Then $\play(s') = \play(t)$, so
$o(\leaf(s')) = o(\leaf(t))$, and hence $s'_i$ is also a best response
of player $i$ to $s_{-i}$.  Since $s_i$ is not a best response to
$s_{-i}$, for some $v$ from $play(s')$ such that $\turn(v) = i$
player's $i$ strategy $s^v_i$ is not a best response to $s^v_{-i}$ in
the subgame $G^v$.  Take the last such node $v$ from $play(s')$.

So $R(v,s) = 0$ while $R(v,s') = 1$, since $(s'_i)^v$ is a best
response to $s^v_{-i}$ in the subgame $G^v$.  Further, by the choice
of $v$ and the definition of $s'_i$ we have that $R(w,s') = R(w,s)$
for all nodes $w$ that precede $v$ on the list $L$.  We conclude that
$P(s) <_{lex} P(s')$.
\HB

\section{Win or lose and chess-like games}
\label{sec:comp}

%In this section we discuss various classes of finite strictly
% competitive extensive games that we introduced in Section
% \ref{sec:prelim}.  Theorem \ref{thm:TDI} shows that in every finite
% extensive game that satisfies the TDI condition all subgame perfect
% equilibria are payoff equivalent.

% \sblue{KEEP?}

% For the subclass of strictly
% competitive extensive games a stronger result holds, namely that all
% Nash equilibria are payoff equivalent.  In fact, by virtue of the
% Minimax Theorem \ref{thm:scminimax} this result already holds for all
% strictly competitive games, not necessarily finite extensive ones.
% The difference is that now it refers to games in which Nash equilibria
% do exist.

In this section we discuss two classes of zero-sum extensive games
introduced in Section \ref{sec:prelim}. We begin with the win or lose
games.  Given such a game $G$ we say that a strategy $s_i$ of player
$i$ is \bfe{winning} if
\[
  \fa s_{-i} \in S_{-i} : o_i(\leaf(s_i, s_{-i})) = 1,
\]
and denote the (possibly empty) set of such strategies
by $\win_i(G)$.

The Matching Pennies game shows that in strategic win or lose games
winning strategies may not exist.
% A classic result, attributed to Zermelo \cite{Zer13}, implies that f
For finite win or lose extensive games the situation changes.

\begin{theorem} \label{thm:win}
  Let $G$ be a finite win or lose extensive game. For all players $i$ we
  have $\win_i(G) \neq \emptyset$ iff $\win_{-i}(G) = \emptyset$.
\end{theorem}

\begin{proof}

  Call the players white and black and call a finite win or lose
  extensive game \emph{white} if the white player has a winning
  strategy in it and analogously for \emph{black}.  We prove that
  every such game is white or black. Clearly, these alternatives are
  mutually exclusive.

  We proceed by induction on the number of nodes in the game tree. The
  claim clearly holds when the game tree has just one node. Consider a
  game $G$ with the game tree $T$ with more than one node.  By the
  induction hypothesis for every child $u$ of the root of $T$ the
  subgame $G^u$ is white or black.

  Without loss of generality assume that in $G$ the white player moves
  first.  We claim that the game $G$ is black if for every child $u$
  of the root of $T$ the subgame $G^u$ is black and otherwise that it
  is white. Indeed, in the first case no matter what is the first move
  of the white player he loses the game if the black player pursues
  his winning strategy in the resulting subgame, and otherwise the
  white player wins the game if he starts by selecting the move
  that leads to a white subgame and subsequently pursues in this
  subgame his winning strategy.

Note that we did not assume that the players alternate their
moves.
\end{proof}

\II

Next we consider chess-like games.
We say that a strategy $s_i$ of
player $i$ in such a game $G$ \bfe{guarantees him at least a draw} if
\[
\fa s_{-i} \in S_{-i} : o_i(leaf(s_i, s_{-i})) \geq 0,
\]
and denote the (possibly empty) set of such strategies 
by $\draw_{i}(G)$.
The set $\win_{i}(G)$ is defined as above.

\begin{theorem} \label{thm:chess}
  Let $G$ be a finite chess-like extensive game. We have
  \[
    \mbox{$\win_{1}(G) \neq \ES$ or $\win_{2}(G) \neq \ES$ or
      ($\draw_{1}(G) \neq \ES$ and $\draw_{2}(G) \neq \ES$).}
  \]
\end{theorem}

% It states that in every chess-like game either one of the players has
% a winning strategy or each player has a strategy that guarantees him
% at least a draw.  These three alternatives are mutually exclusive,
% since for all $i \in \{1,2\}$, $\win_{i}(G) \neq \ES$ implies both
% $\win_{-i}(G) = \ES$ and $\draw_{-i}(G) = \ES$.

We reproduce a proof given in \cite{AS21}.
\begin{proof}
  We introduce the following abbreviations:

  \begin{itemize}
  \item $W_1$ for $\win_{1}(G) \neq \ES$,
    
  \item $D_2$ for $\draw_{2}(G) \neq \ES$,

  \item $W_2$ for $\win_{2}(G) \neq \ES$,

  \item $D_1$ for $\draw_{1}(G) \neq \ES$.
    
  \end{itemize}

Let $G_1$ and $G_2$ be the modifications of $G$ in which each
outcome $(0, 0)$ is replaced for $G_1$ by $(-1,1)$ and for $G_2$
by $(1,-1)$.  Then $\win_{1}(G_1) = \win_{1}(G)$, $\win_{2}(G_1) = \draw_{2}(G)$,
    $\win_{1}(G_2) = \draw_{1}(G)$, and $\win_{2}(G_2) = \win_{2}(G)$.

    Hence by Theorem \ref{thm:win} applied to the games $G_1$ and
    $G_2$ we have $W_1 \lor D_2$ and $W_2 \lor D_1$, so
    $(W_1 \land W_2) \lor (W_1 \land D_1) \lor (D_2 \land W_2) \lor (D_2 \land D_1)$,
    which implies $W_1 \lor W_2 \lor (D_2 \land D_1)$, since
    $\neg (W_1 \land W_2)$, $(W_1 \land D_1) \equiv W_1$, and
    $(D_2 \land W_2) \equiv W_2$.
\end{proof}

\II

The above result is attributed to \cite{Zer13}.  However, in
\cite{SW01} it was pointed out that the paper contains only the idea and
the corresponding result is not formally stated, and that the first
rigorous statement of the result and its proof seems to have been
provided in \cite{Kal28}.  This result is stated in \cite[page
125]{vNM04} and proved using backward induction (apparently the first
use of it in the literature on game theory).  In \cite{Ewe02twoplayer}
a proof is provided that does not rely on backward induction and
argument also covers chess-like games in which infinite plays,
interpreted as draw, are allowed.  As noticed in \cite{AS21} Theorems
\ref{thm:win} and \ref{thm:chess} also hold for infinite extensive
games in which every play is finite.

\section{Conclusions}
\label{sec:conc}

The aim of this tutorial was to provide a self-contained introduction
to finite extensive games with perfect information aimed at computer scientists.
Our objective was to provide a systematic presentation of the most
important results concerning this class of games that in our view
could be of interest to computer scientists.

In \cite{AS21} we argued that the next most natural class of extensive
games is the one in which every play is finite (in the set theory
terminology the game trees are then \emph{well-founded}). In such a
class of games one can in particular consider \emph{behavioural
  strategies} in the extensive games considered here, according to
which a move consists of a probability distribution over the finite
set of children of a given node.

Many textbooks on game theory rather choose as the next class
extensive games with imperfect information.  In these games players do
not need to know the previous moves made by the other players. An
example is the Battleship game in which the first move for each player
consists of a secret placing of the fleet on the grid. An interested
reader is referred to Part III of \cite{OR94}.

\subsection*{Acknowledgements}

We would like to thank Ruben Brokkelkamp, Marcin Dziubi\'{n}ski,
R.~Ramanujam, and an anonymous referee for useful comments on the
first version of this paper.  The second author was partially
supported by the grant MTR/2018/001244.

\bibliographystyle{plain}
\bibliography{ref-s,e}

\begin{thebibliography}{10}

\bibitem{AG11}
K.~R. Apt and E.~Gr{\"{a}}del, editors.
\newblock {\em Lectures in Game Theory for Computer Scientists}.
\newblock Cambridge University Press, 2011.

\bibitem{AS21}
K.R. Apt and S.~Simon.
\newblock Well-founded extensive games with perfect information.
\newblock In {\em Proceedings 18th Conference on Theoretical Aspects of
  Rationality and Knowledge, {TARK} 2021}, volume 335, pages 7--21. Electronic
  Proceedings in Theoretical Computer Science ({EPTCS}), 2021.

\bibitem{Aum91}
R.~Aumann.
\newblock Backward induction and common knowledge of rationality.
\newblock {\em Games and Economic Behavior}, 8:6--19, 1991.

\bibitem{Bat97}
P.~Battigalli.
\newblock On rationalizability in extensive games.
\newblock {\em Journal of Economic Theory}, 74:40--61, 1997.

\bibitem{Dut01}
P.K. Dutta.
\newblock {\em Strategies and Games}.
\newblock MIT Press, Cambridge, MA, 2001.

\bibitem{Ewe02twoplayer}
C.~Ewerhart.
\newblock Backward induction and the game-theoretic analysis of chess.
\newblock {\em Games and Economic Behaviour}, 39:206--214, 2002.

\bibitem{Ewe02}
C.~Ewerhart.
\newblock Iterated weak dominance in strictly competitive games of perfect
  information.
\newblock {\em Journal of Economic Theory}, 107(2):474--482, 2002.

\bibitem{Hal01}
J.~Halpern.
\newblock Substantive rationality and backward induction.
\newblock {\em Games and Economic Behavior}, 37(2):425--435, 2001.

\bibitem{JR01}
G.A. Jehle and P.J. Reny.
\newblock {\em Advanced Microeconomic Theory}.
\newblock Addison Wesley, New York, NY, second edition, 2001.

\bibitem{Kal28}
L.~Kalm\'{a}r.
\newblock Zur theorie der abstrakten spiele.
\newblock {\em Acta Universitatis Szegediensis/Sectio Scientiarum
  Mathematicarum}, 4:65--85, 1928/29.

\bibitem{Kuhn50}
H.~W. Kuhn.
\newblock Extensive games.
\newblock {\em Proc. of the National Academy of Sciences}, 36:570--576, 1950.

\bibitem{Kuk02}
N.S. Kukushkin.
\newblock Perfect information and potential games.
\newblock {\em Games and Economic Behavior}, 38(2):306--317, 2002.

\bibitem{MS97}
L.M. Marx and J.M. Swinkels.
\newblock Order independence for iterated weak dominance.
\newblock {\em Games and Economic Behaviour}, 18:219--245, 1997.

\bibitem{Mil96}
I.~Milchtaich.
\newblock Congestion games with player-specific payoff functions.
\newblock {\em Games and Economic Behaviour}, 13:111--124, 1996.

\bibitem{Mil13}
I.~Milchtaich.
\newblock Schedulers, potentials and weak potentials in weakly acyclic games.
\newblock Working paper 2013-03, Bar-Ilan University, Department of Economics,
  2013.

\bibitem{OR94}
M.J. Osborne and A.~Rubinstein.
\newblock {\em A Course in Game Theory}.
\newblock The MIT Press, 1994.

\bibitem{Ros81}
R.~Rosenthal.
\newblock Games of perfect information, predatory pricing and the chain-store
  paradox.
\newblock {\em Journal of Economic Theory}, 25(1):92--100, 1981.

\bibitem{Ros73}
R.~W. Rosenthal.
\newblock A class of games possessing pure-strategy {Nash} equilibria.
\newblock {\em International Journal of Game Theory}, (2):65--67, 1973.

\bibitem{SW01}
U.~Schwalbe and P.~Walker.
\newblock Zermelo and the early history of game theory.
\newblock {\em Games and Economic Behavior}, 34(1):123--137, 2001.

\bibitem{Sel65}
R.~Selten.
\newblock Spieltheoretische behandlung eines oligopolmodells mit
  nachfragetr\"agheit.
\newblock {\em Zeitschrift f\"ur die gesamte Staatswisenschaft}, 121:301--324
  and 667--689, 1965.

\bibitem{Sta96}
R.~Stalnaker.
\newblock Knowledge, belief and counterfactual reasoning in games.
\newblock {\em Economics and Philosophy}, 12(2):133--163, 1996.

\bibitem{vNM04}
J.~{von Neumann} and O.~Morgenstern.
\newblock {\em Theory of Games and Economic Behavior (60th Anniversary
  Commemorative Edition)}.
\newblock Princeton Classic Editions. Princeton University Press, 2004.

\bibitem{You93}
H.~Peyton Young.
\newblock The evolution of conventions.
\newblock {\em Econometrica}, 61(1):57--84, 1993.

\bibitem{Zer13}
E.~Zermelo.
\newblock {\"Uber} eine anwendung der mengenlehre auf die theorie des
  schachspiels.
\newblock In {\em Proc. of The Fifth International Congress of Mathematicians},
  pages 501--504. Cambridge University Press, 1913.

\end{thebibliography}

\end{document}